\documentclass[12pt,leqno]{amsart}

\usepackage{palatino,mathrsfs}
\usepackage[mathcal]{euler}

\usepackage{xcolor}

\usepackage{graphicx} 
\usepackage{epstopdf,epsfig}

\usepackage{comment} 
\usepackage{amsmath,amsthm,amsfonts,amssymb,latexsym,amscd,enumerate,url,hyperref}
\usepackage{amssymb}
\usepackage{url}
\usepackage{graphicx}
\usepackage[all,knot]{xy}
\xyoption{arc}

\usepackage{chngcntr}
  \counterwithin{equation}{section}

\usepackage{multirow}

\theoremstyle{plain}
\newtheorem{theorem}{Theorem}

 \newtheorem{fact}{Fact}
\newtheorem{lemma}{Lemma}
\newtheorem{proposition}{Proposition}

\theoremstyle{definition}
\newtheorem{example}{Example}
\theoremstyle{remark}
\newtheorem{remark}{Remark}
\newtheorem{conjecture}{Conjecture}

\def\X{\mathbb X}
\def\R{\mathbb R}
\def\C{\mathbb C}

\def\Z{\mathbb Z}

\def\R{\mathbb R}
\def\C{\mathbb C}

\def\Z{\mathbb Z}

\def\BK{\boldsymbol{K}}
\def\BG{\boldsymbol{G}}
\def\g{\boldsymbol{\mathfrak{g}}}

\def\X{\mathbb{X}}
\def\fg{\mathfrak{g}}

\def\F{\mathcal{F}}

\begin{document}

\title{Symmetries of the Hydrogen atom and algebraic families}
 
 \author{Eyal Subag}
 \address{Department of Mathematics, Penn State University, University Park, PA 16802, USA.}
\email{eus25@psu.edu}

\maketitle

\par{\centering \textit{To the memory of Joseph L. Birman}\par}



\begin{abstract}
\noindent  We show how  the Schr\"{o}dinger equation for the hydrogen atom in two dimensions gives rise to an algebraic family of Harish-Chandra pairs that codifies  hidden symmetries. The hidden symmetries vary  continuously    between $SO(3)$, $SO(2,1)$ and the Euclidean group $O(2)\ltimes \R^2$.  We  show that solutions of the Schr\"{o}dinger equation may be organized into an algebraic family of Harish-Chandra modules.  Furthermore, we use Jantzen filtration techniques  to algebraically recover the spectrum of the  Schr\"{o}dinger operator.  This is  a first application to physics of  the algebraic families of Harish-Chandra pairs and modules developed in the work of Bernstein et al. \cite{Ber2016,Ber2017}.  
\end{abstract}

\setcounter{tocdepth}{2}

\tableofcontents

\section{Introduction and main results} 
The  Schr\"{o}dinger equation of the hydrogen atom has  a visible  $O(3)$ symmetry.  Almost a century ago, in the case of  bound states (negative energy), Pauli discovered  new invariants,  independent of the angular momentum, and used them  to derive  the spectrum of the  hydrogen atom  \cite{Pauli}.  A few years later,  Fock  showed that an $SO(4)$ symmetry, known as \textit{the hidden symmetry}, governs the degeneracy of the energy eigenspaces \cite{Fock:1935vv}.   Bargmann explained how the works of Pauli and Fock are related  \cite{Bargmann}. Later on,  it was shown that for scattering states (positive energy) there is an $SO_0(3,1)$ hidden symmetry, and for zero energy there     is an $SO(3)\ltimes \R^3$ (the Euclidean group) hidden symmetry. A similar analysis was carried out  in any dimension $n\geq 2$ of the configuration space  \cite{Bander1,Bander2}.  

The purpose of this paper is to show that an algebraic family of Harish-Chandra pairs can be associated to the Schr\"{o}dinger equation, and  the hidden symmetries  (for all  possible energy values)  can be obtained from the algebraic family. We will show that the parameter space for the family  naturally contains the spectrum of  the Schr\"{o}dinger operator and can be thought of as the space of all ``pseudo-energies".  We further show that the collection of all   physical solutions of the   Schr\"{o}dinger equation (for all  possible energy values)  arises from an algebraic family of Harish-Chandra modules.  Here we shall only consider the  case of the two-dimensional system (see \cite{Cisneros69,Torres1998,Parfitt2002} and reference therein); the general case will be considered elsewhere. 
It should be noted that the symmetries and group theoretical aspects  of the hydrogen atom system, as well as  its classical analogue, the Kepler-Coulomb system, have been   studied extensively  throughout the years, e.g., see \cite{GuilleminandSternberg,Singer,Wulfman}. We shall now describe the setup and results more carefully. 
The  Schr\"{o}dinger equation of the hydrogen atom in $n$ dimensions   ($n$-dimensional configuration space) with $n\geq 2$ is given by 
\begin{eqnarray}\nonumber
&&H\psi(x)=E\psi(x),\\ \nonumber
&& H=-\frac{\hbar^2}{2\mu}\triangle-\frac{k}{r},
\end{eqnarray}
where 
$x=(x_1,x_2,...,x_n)$  is the coordinate vector, $r=\sqrt{\sum_{i=1}^{n}x_i^2}$, $\mu$ is the reduced mass, $k$ is a positive constant, $\hbar$ is the reduced Planck's constant, and $\triangle$ is the Laplacian  $\sum_{i=1}^{n}\frac{\partial^2}{\partial x_i^2}$.
It has an obvious $O(n)$-symmetry but also   less obvious  larger symmetry groups, its  hidden symmetries, which we shall now describe. It is known   
that on  every eigenspace of $H$ with eigenvalue $E$,   the components of the angular momentum vector operator and  the components of the quantum Laplace-Runge-Lenz vector operator 
generate a Lie algebra, $\mathcal{G}_E$, such that  
\[\mathcal{G}_E\simeq \begin{cases}
\mathfrak{so}(n,1), & E> 0,\\
\mathfrak{so}(n)\ltimes \mathbb{R}^{n},& E=0,\\
\mathfrak{so}(n+1), & E< 0,
\end{cases} \]
e.g., see \cite{Bander1,Bander2}.  These Lie algebras are the \textit{ infinitesimal hidden symmetries}. Moreover, the spectrum of $H$  is composed of three disjoint 
pieces, $$\operatorname{Spec}(H)=\mathcal{E}=\mathcal{E}_b\sqcup \mathcal{E}_0\sqcup \mathcal{E}_{s},$$
where: $\mathcal{E}_b$ corresponds to  bound states and  is
 given explicitly by  $\mathcal{E}_b=\big\{E_n= -\frac{C}{(n_0+n)^2}|n=0,1,2,...\big\}$ for some positive constants  $C$ and $n_0$ (see \cite[Eq. 4.115]{Adams} for the exact formula);    $\mathcal{E}_{s}$ corresponds to  scattering states and  is given explicitly by  $\mathcal{E}_s=(0,\infty)$; and   $\mathcal{E}_0=\{0\}$. 
For any $E$ in the spectrum, the physical solution space of the Schr\"{o}dinger equation, $Sol(E)$, is invariant under the  corresponding  hidden symmetry group.  By  physical solutions  we mean  
solutions for the Schr\"{o}dinger equation that are regular at the origin and have best possible behavior at infinity, see Section  \ref{apen}. These solutions were obtained by physicists long ago. For negative energies they are square-integrable but otherwise they are not. 
For $E\in \mathcal{E}_b$ the space  $Sol(E)$,  carries a unitary irreducible representation of $SO(n+1)$; for $E\in \mathcal{E}_s$ it carries a unitary irreducible (principal series)  representation of $SO_0(n,1)$; and for $E=0$ it carries a unitary representation of the Euclidean group $SO(n)\ltimes \R^n$. In general, it is not known if the representation of $SO(n)\ltimes \R^n$ is irreducible,  but for $n=2$ it is (this  follows from \cite{Torres1998}).

So far we have reviewed   known facts. To explain the novelty of this work we shall first discuss algebraic families  of Harish-Chandra pairs and modules. The idea of \textit{contraction} of Lie groups and their representations  is  well known in  the  mathematical physics literature, see e.g., \cite{Inonu-Wigner53,Saletan61,Dooley85,Gilmore05,Subag12}. Recently, it was demonstrated how algebraic families  serve  a good mathematical framework for contractions  \cite{Ber2016,Ber2017}. We will show here that  energy is a natural deformation parameter (or contraction parameter) for the hydrogen atom system   that fits perfectly with the theory of algebraic families that was introduced in \cite{Ber2016,Ber2017}.

Roughly speaking, an algebraic family of complex Lie  algebras over a complex algebraic variety $X$ is  a collection  of complex  Lie algebras $\g=\{\g|_x\}_{x\in X}$  that vary algebraically in $x$. We shall refer to  $\g|_x$ as the \emph{fiber} of $\g$ at $x\in X$. Similarly, there is a natural notion  for an algebraic family of complex algebraic groups over $X$, $\BK=\{\BK|_x\}_{x\in X}$. 
An algebraic family of Harish-Chandra pairs  over $X$ is a pair $(\g,\BK)$, where $\g$ is an algebraic  family of complex Lie algebras  and $\BK$ is a compatible algebraic family of complex algebraic groups, both  over the same base $X$. In  Section \ref{sec3} we shall see that in the case of the two-dimensional hydrogen atom, the Schr\"{o}dinger operator $H$ gives rise to an algebraic family of Harish-Chandra pairs  $(\g,\BK)$  over  $X:=\C$, where $\BK$ is the constant family of groups $\C\times O(2,\C)$ and for any  $x\in \X$ the fibers of $\g$ satisfy  
\[ \g|_x\simeq \begin{cases}
\mathfrak{so}(3,\C) & x\neq 0,\\
 \mathfrak{so}(2,\C)\ltimes \C^2 & x=0.
\end{cases}\] 
The physical interpretation of the family gives rise to  a canonical real structure $\sigma$ on   $(\g,\BK)$,  (see Section \ref{sec3.3}) which leads to a family of real Harish-Chandra pairs   $(\g^{\sigma},\BK^{\sigma})$ over  $\X^{\sigma}=\R$ with $\BK^{\sigma}=\R\times O(2)$ and,  for $x\in \R$,
\[ \g^{\sigma}|_x\simeq  \begin{cases}
\mathfrak{so}(2,1), & x> 0,\\
 \mathfrak{so}(2)\ltimes \R^2, & x=0,\\
\mathfrak{so}(3), & x< 0.
\end{cases} \]
Moreover, by  construction, points of $\X$ correspond to ``generalized" eigenvalues of $H$.  This allows us to regard $\operatorname{Spec}(H)$ as a subset of  $\X^{\sigma} $. We will prove the following theorem. \\

\noindent \textbf{Theorem III.1.}
\textit{For any $E\in \operatorname{Spec}(H) \subset \X^{\sigma}$ the visible symmetry of the  Schr\"{o}dinger equation $H\psi=E\psi$   is  $\BK|_E^{\sigma}$, and the infinitesimal hidden symmetry is  $\g|_E^{\sigma}$. Furthermore,   $\g^{\sigma}$ can be lifted to a family of Lie groups  that  correspond to the  hidden symmetries. That is,  there is an algebraic family of complex algebraic groups $\BG$ over $\X$ with a real structure $\sigma_{\BG}$,  such that, for every $E\in \X^{\sigma}$,
 \[\BG|_E^{\sigma_{\BG}}\simeq  \begin{cases}
SO(2,1), & E>0,\\
O(2)\ltimes \R^2, & E=0,\\
SO(3), & E< 0.
\end{cases}\] }\\

 There is  an obvious notion of an algebraic family  of  Harish-Chandra modules for   $(\g,\BK)$  (also known as   an algebraic  family of $(\g,\BK)$-modules); see \cite[Sec. 2.4]{Ber2016}. In simple  cases, like the examples considered here, an algebraic family of Harish-Chandra modules for $(\g,\BK)$ is an algebraic vector bundle  $\F$ (of infinite rank) over $\X$ which carries compatible actions of $\g$ and  $\BK$. In particular, for $x\in \X$, the fiber of $\F$ at $x$ is a  $(\g|_x,\BK|_x)$-module. The family $\g$ has a natural counterpart  to the   Casimir element of a semisimple Lie algebra, which we shall call  the \textit{regularized Casimir} $\Omega$ (see Section \ref{sec4.1}). On generically irreducible families of $(\g,\BK)$-modules, $\Omega$ must act via multiplication by some  polynomial function on $\X$. In Section \ref{sec4.2} we show that the physical realization of $\g$ forces $\Omega$ to act on $\F|_E$ via multiplication by $\omega(E)=-\frac{E}{4}-\frac{k^2}{2}$. \\ 
 
 \noindent \textbf{Theorem IV.1.}
\textit{Let  $\F$ be a generically irreducible and quasi-admissible  family of $(\g,\BK)$-modules  on which $\Omega$ acts via multiplication by  $\omega(E)=-\frac{E}{4}-\frac{k^2}{2}$.   The collection of all the reducibility points of $\F$ coincides with $\mathcal{E}_b$.}\\

Using the real structure of $(\g,\BK)$, to   any  family $\F$ of Harish-Chandra modules  one can associate  a dual family $\mathcal{F}^{\langle\sigma\rangle}$ (the $\sigma$-twisted dual, see Section \ref{sec4.4.1} and \cite[Sec. 2.4 ]{Ber2017}). A nonzero morphism of   $(\g,\BK)$-modules from $\F$ to $\mathcal{F}^{\langle\sigma\rangle}$ (an intertwining operator) induces a canonical filtration on every fiber of $\F$. This is the Jantzen filtration; see Section \ref{sec4.4} and \cite[Sec. 4]{Ber2017}. The Jantzen filtration gives a new algebraic way to calculate the spectrum of $H$,  and the subset of the spectrum corresponding to bound states. In Section \ref{sec4.3} we  show that there are exactly two families of generically irreducible and quasi-admissible  families of $(\g,\BK)$-modules  such that:  $\Omega$ acts via multiplication by  $\omega(E)=-\frac{E}{4}-\frac{k^2}{2}$;  and that are  generated by their isotypic subsheaf,  $\F_0 $, corresponding to the trivial $SO(2,\C)$-type. We  prove the following.  \\

\noindent \textbf{Theorem IV.2.}
\textit{Let  $\F$ be one of the two generically irreducible and quasi-admissible  families of $(\g,\BK)$-modules on which $\Omega$ acts via multiplication by  $\omega(E)=-\frac{E}{4}-\frac{k^2}{2}$,  and that are generated by $\F_0$. The spectrum of $H$ coincides with the set of all  $E\in \X$ for which $\F|_E$  has  a nonzero infinitesimally unitary Jantzen quotient. Moreover,  $\mathcal{E}_b$   coincides with the set of all  $E\in \X$ for which $\F|_E$ has a nontrivial Jantzen filtration.}\\

For such a family $\F$, we  show that for each $E\in \operatorname{Spec}(H)$  the fiber $\F|_E $ has  exactly one infinitesimally unitary Jantzen quotient   $J(\F|_E)$. Furthermore, we prove the following result.\\

\noindent \textbf{Theorem IV.3.}
\textit{Let  $\F$ be as in Theorem \ref{th3}. Then for any $E\in  \operatorname{Spec}(H)$, the infinitesimally unitary Jantzen quotient $J(\F|_E)$ can be integrated to the unitary irreducible representation of the connected component of $\BG|_E^{\sigma_{\BG}}$ which is isomorphic to $Sol(E)$.}\\

In the discussion  section we speculate on  relations between $\F$,  the measurable field of Hilbert spaces arising from the spectral theory for $H$, and solution spaces for the Schr\"{o}dinger equation. We conjecture that for almost every energy $E$,  the Janzten quotient $J(\F|_E)$, the $SO(2)$-finite physical solutions in $Sol(E)$,  and  the $SO(2)$-finite vectors in the Hilbert space $\mathcal{H}_E$ appearing in the spectral decomposition, can be identified in a canonical way.

This project  was  initiated after stimulating  discussions with Joseph L. Birman.  This work is dedicated to his memory. The author thanks Moshe Baruch, Joseph Bernstein, Nigel Higson,  Ady Mann, and Willard Miller for many useful discussions.   The author  thanks the anonymous referee for  pointing out  the possible relation between our algebraic families and the measurable field associated with the  Schr\"{o}dinger operator.  The advice of Nigel Higson on the part dealing with  measurable fields is greatly appreciated.

 \section{Algebraic families}
 In this section we shall review the  formalism of algebraic families of Harish-Chandra pairs and their  modules.  We shall  avoid technicalities and convey some ideas by examples. More details can be found in   \cite{Ber2016,Ber2017}. Algebraic families can be defined over any  complex algebraic variety. For our purpose, it is enough to consider families over the simplest non-trivial affine variety,  $\mathbb{A}^1_{\C}=\C$. So, throughout this note we let $\X$ be the complex affine  line, i.e.,  $\X=\C$. As usual, we shall denote the structure sheaf of regular functions of  $\X$ by $O_{\X}$. Moreover, we shall freely identify  various families (various sheaves) over $\X$ with their space of global sections.  
 \subsection{Families of complex  Lie algebras}
An algebraic family of complex Lie algebras $\g$ over $\X$ is a locally free sheaf of $O_{\X}$-modules equipped with an $O_{\X}$-bilinear Lie bracket. Since $\X$ is affine, such a family is nothing else  but a Lie algebra over the ring  $\C[x]=O_{\X}(\X)$. Intuitively, $\g$ should be thought of as a collection of complex Lie algebras (the fibers of $\g$) parameterized by $\X$ that vary algebraically.  Recall that the fiber of $\g$ at $x_0\in \X$ is  $\g|_{x_0}:=\g/I_{x_0}\g$,  where $I_{x_0}$  is the maximal ideal of  $\C[x]$ consisting of  all functions the vanish at $x_0$. Of course,  $\g|_{x_0}$  is a  Lie algebra over $\C\simeq \C[x]/I_{x_0}$.

\begin{example}
Consider the constant family of Lie algebras over $\X$ with fiber $\mathfrak{gl}_3(\C)$,  that is,  the sheaf  of regular (algebraic) sections of the bundle  $\X\times\mathfrak{gl}_3(\C)$ over $\X$. This is  an  algebraic family of Lie algebras over $\X$. We shall denote it by $\boldsymbol{\mathfrak{gl}_3(\C)}$. Each fiber of $\boldsymbol{\mathfrak{gl}_3(\C)}$ is canonically identified with ${\mathfrak{gl}_3(\C)}$. The family $\boldsymbol{\mathfrak{gl}_3(\C)}$ contains interesting non-constant subfamilies. We shall now describe one such subfamily that will play a role in what follows. Let $\widetilde{\boldsymbol{\mathfrak{so}_3}}$ be the subfamily that is characterized by the following property. For every $x\in \X$, the fiber $\widetilde{\boldsymbol{\mathfrak{so}_3}}|_x$ is given (under the above mentioned identification) by 
\[ \left\{\left(\begin{matrix}
0 & \alpha & \beta \\
-\alpha  &0 &\gamma\\
 -x\beta &  -x \gamma& 0
\end{matrix}\right)\middle| \alpha,\beta,\gamma \in \C  \right\}. \]  
Note that \[ \widetilde{\boldsymbol{\mathfrak{so}_3}}|_x\simeq \begin{cases}
\mathfrak{so}(3,\C) & x\neq 0,\\
 \mathfrak{so}(2,\C)\ltimes \C^2 & x=0.
\end{cases}\] 
Let $e_{ij}$ with $i,j\in \{1,2,3\}$ be the standard basis of $\mathfrak{gl}_3(\C)$. The maps (from $\C$ to $\mathfrak{gl}_3(\C)$) given by 
\begin{eqnarray}\nonumber
&&j_1(x)=  e_{23}-x e_{32}, \\ \nonumber
&&j_2(x)=e_{13}-x e_{31}, \\ \nonumber
&& j_3(x)=e_{12}-e_{21},
\end{eqnarray}
define a basis for $\widetilde{\boldsymbol{\mathfrak{so}_3}}$ as a Lie algebra over $\C[x]$. In particular, 
\[\widetilde{\boldsymbol{\mathfrak{so}_3}}=\C[x]j_1\oplus \C[x]j_2\oplus \C[x]j_3. \]
The commutation relations are determined by 
\begin{eqnarray}\nonumber
&&[j_1(x),j_2(x)]=xj_3(x),  [j_2(x),j_3(x)]=j_1(x), [j_3(x),j_1(x)]=j_2(x). 
\end{eqnarray}
\end{example}
 \subsection{Families of  complex algebraic groups}  
Formally, an algebraic family of complex algebraic groups  is a smooth group scheme $\BG$ over $\X$ with  $\BG$ being a  smooth complex algebraic variety,  see \cite[Sec. 2.2]{Ber2016}. As in the case of families of Lie algebras, we can think about it as a collection of complex algebraic groups that vary algebraically.  For us, the most important example is the  constant family of groups over $\X$ with fiber $GL_3(\C)$. We shall denote it by $\boldsymbol{GL_3(\C)}$.  Any fiber of $\boldsymbol{GL_3(\C)}$ is canonically identified with  ${GL_3(\C)}$. We shall now define a subfamily $\widetilde{\boldsymbol{SO_3}}$ of $\boldsymbol{GL_3(\C)}$, whose associated family of Lie algebras (\cite[Sec. 2.2.1]{Ber2016}) is  $\widetilde{\boldsymbol{\mathfrak{so}_3}}$ from the previous section. 
\begin{example}
The family  $\widetilde{\boldsymbol{SO_3}}$ is uniquely determined by its fibers  which we shall now describe. 
For a nonzero $x\in \X$, $\widetilde{\boldsymbol{SO_3}}|_{x} $ is given by all $A\in SL_3(\C)$ such that $A^tD_xA=D_x$, where $D_x$ is the diagonal matrix in  $ GL_3(\C)$ with diagonal entries $(x,x,1)$. In particular,    $\widetilde{\boldsymbol{SO_3}}|_{x} \simeq SO(3,\C)$.
The remaining fiber is 
\[\widetilde{\boldsymbol{SO_3}}|_{0} =
\left\{\left(\begin{matrix}
A & v \\
0 &|A|
\end{matrix}\right)\middle| A\in O(2,\C), v\in \C^2  \right\} \simeq O(2,\C)\ltimes \C^2.  \]  
To show that $\widetilde{\boldsymbol{SO_3}}$  is indeed an algebraic family of complex algebraic groups one should follow the same  calculations as in  \cite{Bar2017}.
\label{ex2}
\end{example}
\subsection{Families of  Harish-Chandra pairs}  
Before we discuss families of Harish-Chandra pairs we recall the definition of a (classical) Harish-Chandra pair. 
A  Harish-Chandra pair consists of a pair $(\fg,K)$, where $\fg$ is a complex Lie algebra and $K$ a complex algebraic group acting on $\fg$, and an embedding of Lie algebras $\iota:\operatorname{Lie}(K)\longrightarrow \fg$ such that:
\begin{enumerate}
\item $\iota$ is equivaraint where $K$  acts on $\operatorname{Lie}(K)$ by the adjoint action.
\item The action of $\operatorname{Lie}(K)$ on $\fg$   coming from the derivative of the action of $K$ coincides with the action that is obtained by composing $\iota$ with the adjoint action of $\fg$ on itself. 
\end{enumerate}
 For more details see, e.g.,  \cite{Bernstein97,KnappVogan}. An 
 \textit{algebraic family of a Harish-Chandra pairs over $\X$} is defined in analogous way. The main difference is the replacement of $\fg$ by a family of Lie algebras $\g$ over $\X$ and replacement of $K$ by an algebraic family of complex algebraic groups $\BK$ over $\X$. We shall only consider the case in which $\BK$ is a constant family over $\X$. For precise definition see   \cite[Sec. 2.3]{Ber2016}. Below is an  example that will be of main interest for us.

\begin{example}
The family  $\widetilde{\boldsymbol{SO_3}}$ contains a constant subfamily $\boldsymbol{O_2}$ with fiber isomorphic to $O(2,\C)$. Explicitly, for any $x\in \X$,     
\[\boldsymbol{O_2}|_{x} =
\left\{\left(\begin{matrix}
A & 0 \\
0 &|A|
\end{matrix}\right)\middle| A\in O(2,\C)  \right\}.  \] 
Note that any matrix in $\boldsymbol{O_2}|_{x}$ is generated by matrices of the form
\[\widetilde{R}(\theta):=\left(\begin{matrix}
{R}(\theta)&0\\
0&1
\end{matrix}\right), \hspace{2mm} \widetilde{s}:=\left(\begin{matrix}
s&0\\
0 &-1
\end{matrix}\right),\] 
where
\begin{eqnarray}\nonumber
&&{R}(\theta):=\left(\begin{matrix}
\cos(\theta) & \sin(\theta)  \\
-\sin(\theta)  &\cos(\theta) 
\end{matrix}\right), \hspace{2mm} s:=\left(\begin{matrix}
0 & 1 \\
1&0
\end{matrix}\right),
\end{eqnarray}
with $\theta \in \C$.
The family $\boldsymbol{O_2}$  naturally acts on $\widetilde{\boldsymbol{\mathfrak{so}_3}}$  via conjugation. This action is determined  by the formulas  
\begin{eqnarray}\nonumber
 \widetilde{R}(\theta) \cdot j_1=&\cos(\theta)j_1-\sin(\theta)j_2, \hspace{2mm}&\widetilde{s} \cdot j_1=-j_2, \\
\nonumber
 \widetilde{R}(\theta) \cdot j_2=&\sin(\theta)j_1+\cos(\theta)j_2, \hspace{2mm}&\widetilde{s} \cdot j_2=-j_1, \\
\nonumber
 \widetilde{R}(\theta) \cdot j_3=&\hspace{-3cm}j_3, &\widetilde{s} \cdot j_3=-j_3. 
\end{eqnarray}
The  family of Lie algebras, $\boldsymbol{\mathfrak{o}_2}$, associated with  $\boldsymbol{O_2}$    coincides with the  subfamily of 
$\widetilde{\boldsymbol{\mathfrak{so}_3}}$  that is generated by $j_3$. All in all, the pair $(\widetilde{\boldsymbol{\mathfrak{so}_3}},\boldsymbol{O_2})$  is an algebraic family of Harish-Chandra pairs over $\X$.
 \label{ex3}
\end{example} 

\subsection{Real structure for  families}\label{sec2.4} 
In this section we discuss real structures of a family of complex groups and show how it gives rise to  a family of Lie groups.

In general, for any complex algebraic variety $X$  we  denote the complex conjugate variety by $\overline X$. The underlying set of $\overline X$ coincides with that of $X$, and a complex valued function on $\overline X$ is regular if (by definition) its complex conjugate is a regular function on $X$. In addition, $\overline{\overline{X}}$ is canonically isomorphic to $X$. An antiholomorphic morphism from   $X$ to another complex algebraic variety  $Y$ is  an algebraic  morphism from $X$ to $\overline{Y}$. Given any morphism $\psi:X\longrightarrow \overline{Y}$, there is a canonical morphism $\overline{\psi}:\overline{X}\longrightarrow Y$ such that $\psi =\overline{\psi}$ as  set theoretic maps.   A \emph{real structure}, or \emph{antiholomorphic involution}, of $X$ is a morphism  $\sigma_X:X\longrightarrow \overline{X}$ such that  the composition 
\[
\xymatrix{
X\ar[r]^{\sigma_X} & \overline{X}\ar[r]^{\overline{\sigma_X}} &  X} 
\]
is the identity, see \cite[Sec. 2.5]{Ber2016}. We denote by  $X^{\sigma}$ the fixed point set of $\sigma$. We treat $X^{\sigma}$  as a topological space  equipped with  the subspace topology as a subset of $X$ with its analytic topology.  We shall mainly work with $X=\C$ and $\sigma$ being the usual complex conjugation. In  this case $X^{\sigma}=\R$.   

If $\BG$ is an algebraic family of complex algebraic groups over $X$, then $\overline{\BG}$    is an algebraic family of complex algebraic groups over $\overline{X}$. A \emph{real structure}, or \emph{antiholomorphic involution}, of the family $\BG$ is a pair of involutions of varieties,  $\sigma_X:X\longrightarrow \overline{X}$   and $\sigma_{\BG}:\BG \longrightarrow \overline{\BG}$, such that the diagram 
\[
\xymatrix{
\BG\ar[d]\ar[r]^{\sigma_{\BG}} & \overline \BG\ar[d]\\
X\ar[r]^{\sigma_X} & \overline{X} } 
\]
is commutative, and it  induces a morphism $\BG \longrightarrow \sigma_X^*\overline{\BG}$ of algebraic families of complex algebraic groups over $X$. See \cite[Sec. 2.5]{Ber2016}.   For each $x\in X^{\sigma}$, the morphism $\sigma_{\BG}$ induces an antiholomorphic involution of the complex algebraic group  $\BG|_x$.  The fixed point set  of this involution, denoted by  $\BG|_x^{\sigma}$, is a Lie group. In this way we get a collection of Lie groups parameterized by  $X^{\sigma}$. We shall refer to this collection as a family of Lie groups over $X^{\sigma}$ and we denote it by $\BG^{\sigma}$.

For our purpose, it is enough to understand the simplest case in which $X$ is the variety $\X=\C$, $\sigma_{\X}$ is complex conjugation of complex numbers, $\BG$ is the constant family $\X\times GL_n(\C)$,    and $\sigma_{\BG}(x,g)=(\overline{x},\overline{g})$ where  $\overline{x}$ is the complex conjugate of $x\in \C$ and   $\overline{g}$ is the complex conjugate of the matrix  $g\in  GL_n(\C)$. In this case the family $\BG^{\sigma}$ is the constant family $\R\times GL_n(\R)$. We will obtain interesting non-constant  families of Lie groups as  subfamilies of    $\R\times GL_n(\R)$, as in the following example. 
\begin{example}\label{ex4}
Consider the family   $\widetilde{\boldsymbol{SO_3}}$ from example \ref{ex2} with the above mentioned real structure.  Then for any $x\in \R$, the group $\widetilde{\boldsymbol{SO_3}}|_{x}^{\sigma}$ is given by all real matrices in  $\widetilde{\boldsymbol{SO_3}}|_{x}$. In particular, $\widetilde{\boldsymbol{SO_3}}|_{1}^{\sigma}=SO(3)$,    $\widetilde{\boldsymbol{SO_3}}|_{-1}^{\sigma}=SO(2,1)$, and
\[\widetilde{\boldsymbol{SO_3}}|_{0}^{\sigma} =
\left\{\left(\begin{matrix}
A & v \\
0 &|A|
\end{matrix}\right)\middle| A\in O(2), v\in \R^2  \right\}.  \]
The isomorphism classes of the fibers are given by  
\[\widetilde{\boldsymbol{SO_3}}|_{x}^{\sigma}\simeq \begin{cases}
 SO(2,1) & x< 0,\\
 O(2)\ltimes \mathbb{R}^2& x=0,\\
 SO(3) & x> 0.
\end{cases}\]
\end{example}  
There are similar  definitions for  a \emph{real structure}, or \emph{antiholomorphic involution} of algebraic families of complex Lie algebras and algebraic families of Harish-Chandra pairs see \cite[Sec. 2.5]{Ber2016}.  In the above example the real structure on $\widetilde{\boldsymbol{SO_3}}$  induces one on  $(\widetilde{\boldsymbol{\mathfrak{so}_3}},\boldsymbol{O_2})$ from example \ref{ex3}. The corresponding family of real Harish-Chandra pairs  is obtained  by considering the  real matrices in each of the fibers $(\widetilde{\boldsymbol{\mathfrak{so}_3}}|_x,\boldsymbol{O_2}|_x)$. Explicitly, for any $x\in \R$,
\[ \widetilde{\boldsymbol{\mathfrak{so}_3}}|_x^{\sigma}=\left\{\left(\begin{matrix}
0 & \alpha & \beta \\
-\alpha  &0 &\gamma\\
 -x\beta &  -x \gamma& 0
\end{matrix}\right)\middle| \alpha,\beta,\gamma \in \R   \right\} \]  and
\[\boldsymbol{O_2}|_{x}^{\sigma} =
\left\{\left(\begin{matrix}
A & 0 \\
0 &|A|
\end{matrix}\right)\middle| A\in O(2)  \right\}.  \] 
In addition, 
\[ \widetilde{\boldsymbol{\mathfrak{so}_3}}^{\sigma}= \R[x]j_1\oplus \R[x]j_2\oplus \R[x]j_3. \] 
For $x\in \R$,  each homogenous space $\widetilde{\boldsymbol{SO_3}}|_{x}^{\sigma}/\boldsymbol{SO_2}|_{x}^{\sigma}$ has  a geometric meaning. For $x>0$,  it  is a two-sheeted hyperboloid.  For $x<0$, it is an ellipsoid,  and as $x$ approaches zero, the space ``flattens" into two parallel planes. See figure \ref{fig12}.

  \begin{center}
\begin{figure}
\includegraphics[width=\textwidth]{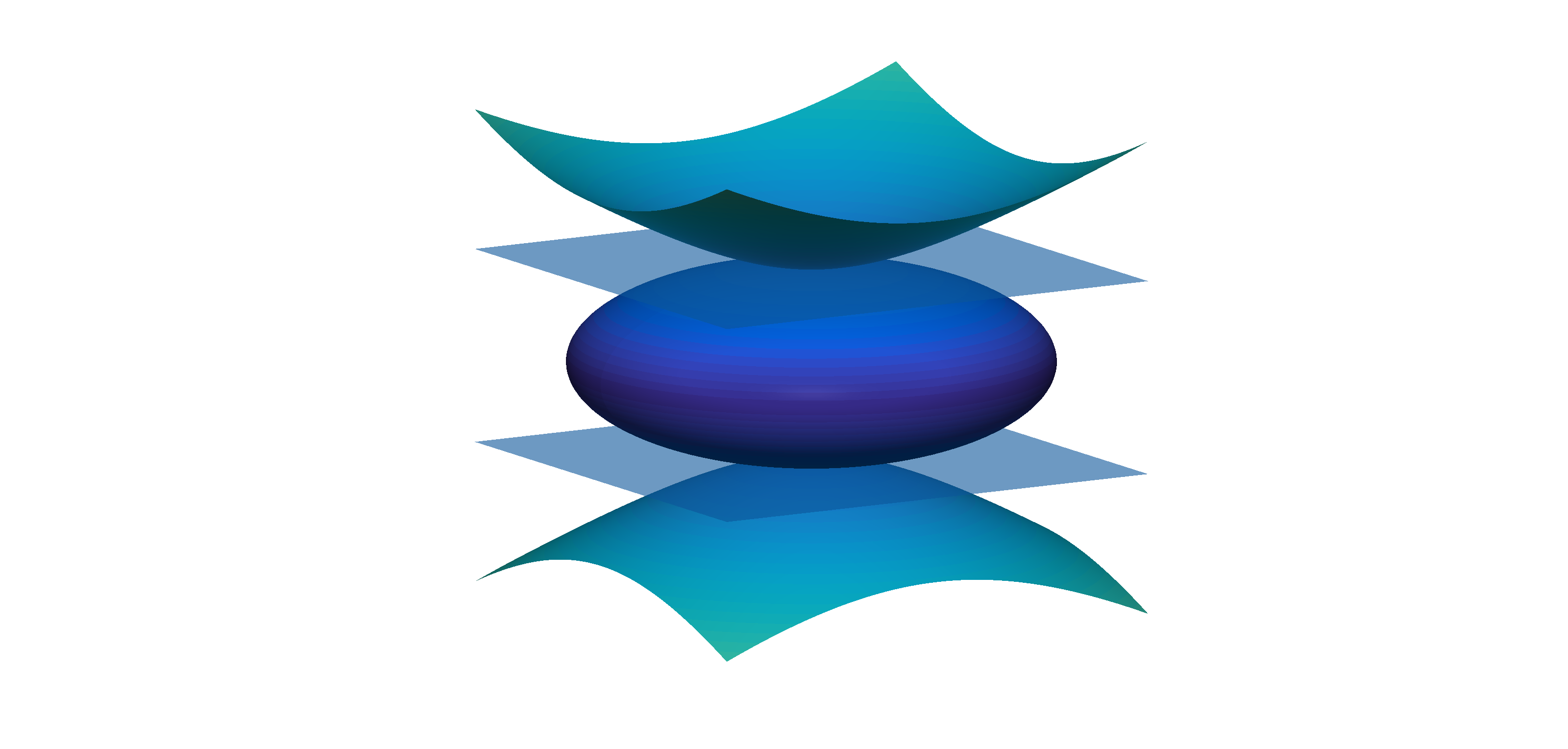}
\caption{The homogeneous spaces $\widetilde{\boldsymbol{SO_3}}|_{x}^{\sigma}/\boldsymbol{SO_2}|_{x}^{\sigma}$}\label{fig12}
\end{figure}
\end{center}


\section{The family of Harish-Chandra pairs  and hidden symmetries }\label{sec3}
In this section we show how the Schr\"{o}dinger equation of the hydrogen atom in two dimensions gives rise to an algebraic family of Harish-Chandra pairs equipped  with a real structure. We then show that the various hidden symmetries can be recovered from the algebraic family.  
\subsection{The infinitesimal hidden symmetry}\label{cent} 
The  Schr\"{o}dinger equation of the hydrogen atom in two dimensions is given by \begin{eqnarray}\label{201}\nonumber
&&H\psi=E\psi,\\ \nonumber
&& H=-\frac{1}{2}(\partial_{xx}+\partial_{yy})-\frac{k}{\sqrt{x^2+y^2}}.
\end{eqnarray}
Here $k$ is a positive constant and  we use atomic units, that is,  $\mu=\hbar=1$.
We note that  $H$ belongs to the  algebra of all smooth differential operators on the space of smooth complex valued functions on $\mathbb{S}=\R^2\setminus \{0\}$. We denote this algebra by $\mathbb{D}$. 
 The algebra $\mathbb{D}$ is naturally filtered by the degree of a differential operator. We denote this increasing filtration by $\mathbb{D}=\cup_{n\geq 0} \mathbb{D}_n$.  Let  $C_{\mathbb{D}}(H)$ be the centralizer of $H$ in $\mathbb{D}$  and $C^2_{\mathbb{D}}(H)$ the intersection of $C_{\mathbb{D}}(H)$ with $ \mathbb{D}_2$. The next lemma follows from a direct calculation. 
\begin{lemma}
The complex vector space  $C^2_{\mathbb{D}}(H)$ is  spanned by:
\begin{eqnarray}\nonumber
&&L:=y\partial_x-x\partial_y, 	L^2, \mathbb{I}\hspace{1mm} (\text{the identity operator}), \\ \nonumber
&& H:=-\frac{1}{2}(\partial_{xx}+\partial_{yy})-\frac{k}{\sqrt{x^2+y^2}},\\ \nonumber
&&   A_y:=\frac{i}{\sqrt{2}}\left(y\partial_{xx}-x\partial_{xy}-\frac{1}{2}\partial_{y}+k\frac{y}{\sqrt{x^2+y^2}}\right),\\ \nonumber
&& A_x:= \frac{i}{\sqrt{2}}\left(x\partial_{yy}-y\partial_{xy}-\frac{1}{2}\partial_{x}+k\frac{x}{\sqrt{x^2+y^2}}\right).
\end{eqnarray}\label{lem1}
\end{lemma}

We shall consider the complex associative sub-algebra of  $C_{\mathbb{D}}(H)$ that  is generated (as an algebra) by $C^2_{\mathbb{D}}(H)$.  We denote it by $\mathcal{B}$.  
\begin{remark}
The algebra $\mathcal{B}$ is the  quantum superintegrable system that is  known as $E18$, see \cite{Miller76,Miller01} and specifically \cite[Sec. 3.1]{Miller13}. It is not known if $\mathcal{B}=C_{\mathbb{D}}(H)$, but there are  indications that this is the case. For example, in the corresponding classical superintegrable system the analogous statement holds. See the paragraph following Definition 10 in  \cite{Miller13}.
\end{remark}
 Observe that $\mathcal{B}$ is in fact an associative algebra over the free polynomial ring $\C[H]$ and it is generated, as an algebra over $\C[H]$,  by  $\{ A_x, A_y, L,\mathbb{I}\}$. The generating set  $\{ A_x, A_y, L,\mathbb{I}\}$  is minimal in two ways. There are no smaller sets of generators and  the sum of  degrees (with respect to the natural filtration of $\mathbb{D}$) of the generators is minimal among all other sets of generators.
 In fact, all such minimal sets of four generators,  with their sum of their degrees equal to $5$,  span the same four dimensional  vector subspace  $V:=\operatorname{span}_{\C}\{ A_x, A_y, L, \mathbb{I}\}$.  Let $\boldsymbol{V}$ be the $\C[H]$-module generated  by  $V$. 
 \begin{lemma}
 $\boldsymbol{V}$  is a Lie algebra over $\C[H]$. That is,  $\boldsymbol{V}$   is an algebraic family of complex  Lie algebras over $\C$.
 \end{lemma}
 \begin{proof}
From the linear independence of $\{ A_x, A_y, L,\mathbb{I}\}$, 
\begin{eqnarray}\nonumber
&& \boldsymbol{V}=\C[H]{L}\oplus\C[H]{A}_x\oplus
\C[H]{A}_y \oplus
\C[H] \mathbb{I}.
\end{eqnarray}
  By direct calculation, 
 \begin{eqnarray}
&&[{A}_y,{L}]={A}_x,[{L},{A}_x]={A}_y, [{A}_x,{A}_y]=-H{L},
\label{cm}
\end{eqnarray}
 and $\mathbb{I}$ commutes with  everything else. 
 \end{proof}
For purposes that will be made clear below, from now on we shall  use the notation $E$  for the indeterminate $H$. Occasionally, and abusing notation, we shall treat the indeterminate $E$ as a point in $\C$, this will be clear from the context. 
For any nonzero  $E\in \C$, the fiber  $\boldsymbol{V}|_E$ is a complex reductive Lie algebra over $\C$ (it is isomorphic to $\mathfrak{gl}_2(\C)$).  Let $\g$ be the  largest Lie subalgebra (over $\C[E])$ of $\boldsymbol{V}$ such that for any nonzero $E\in \C$, the  fiber $\g|_E$ is semisimple. The commutation relations (\ref{cm}) imply that $\{A_x, A_y, L\}$ form a basis for $\g$ over $\C[E]$. 

\begin{proposition}\label{propos1}
$\g$ is isomorphic to $\widetilde{\boldsymbol{\mathfrak{so}_3}}$.
\end{proposition}
\begin{proof}
As mentioned above, 
\begin{eqnarray}\nonumber
&& \boldsymbol{\mathfrak{g}}=\C[E]{L}\oplus\C[E]{A}_x\oplus
\C[E]{A}_y,
\end{eqnarray}
and  
 \begin{eqnarray}
&&[{A}_y,{L}]={A}_x,[{L},{A}_x]={A}_y, [{A}_x,{A}_y]=-E{L}.
\end{eqnarray} \label{204}
The correspondence  
\begin{eqnarray}
&&x\longleftrightarrow-E,\quad  \mathcal{A}_x\longleftrightarrow j_1,\quad \mathcal{A}_y\longleftrightarrow j_2, \quad \mathcal{L} \longleftrightarrow j_3
\label{eq3.4}
\end{eqnarray}
defines an isomorphism between  $\g$ and $\widetilde{\boldsymbol{\mathfrak{so}_3}}$.  
\end{proof}

\begin{remark}
We shall see in Sections \ref{sec3.2} and \ref{sec3.3}  that  the  isomorphism (\ref{eq3.4})  extends to an isomorphism of algebraic families of Harish-Chandra pairs equipped with a real structure. This isomorphism  is essentially unique; all such isomorphisms are given by
\[x\longleftrightarrow-rE,\quad  {A}_x\longleftrightarrow cj_1,\quad {A}_y\longleftrightarrow c^{-1}j_2, \quad {L} \longleftrightarrow j_3,
\] for some positive $r$ and $c$. 
\end{remark}

\subsection{The visible symmetries }\label{sec3.2}

The group $O(2)$  naturally acts on    $\mathbb{S}=\R^2\setminus\{0\}$. This action induces an action of  $O(2,\C)$ on  $\mathbb{D}$ which descends to $\g$ and defines an action of $\BK$, the constant family over $\X$ with fiber $O(2,\C)$. Explicitly, 
\begin{eqnarray}\nonumber
 &{R}(\theta) \cdot {A}_x=\cos(\theta){A}_x-\sin(\theta){A}_y, \hspace{2mm}&s \cdot {A}_x=-{A}_y, \\
\nonumber
 &{R}(\theta) \cdot {A}_y=\sin(\theta){A}_x+\cos(\theta){A}_y, \hspace{2mm}&s \cdot {A}_y=-{A}_x, \\
\nonumber
 &{R}(\theta) \cdot {L}={L},\hspace{37mm} &s \cdot {L}=-{L}. 
\end{eqnarray}
In particular, we see that under the canonical isomorphism ${\boldsymbol{O_2}}\simeq \BK$ the following holds.
\begin{lemma}
The isomorphism \ref{eq3.4}  canonically extends to an isomorphism of algebraic families  of Harish-Chandra pairs over $\X$ between $(\widetilde{\boldsymbol{\mathfrak{so}_3}},\boldsymbol{O_2})$ and $(\g,\BK)$.\qed
\end{lemma}

\subsection{Real structure on $(\g,\BK)$}\label{sec3.3}
 We shall now describe a natural real structure on $(\g,\BK)$. Under the isomorphism $(\g,\BK)\simeq(\widetilde{\boldsymbol{\mathfrak{so}_3}}, \boldsymbol{O_2}) $ the real structure coincides with the one that was given in Example  \ref{ex4}. 
 
The natural inner product on $L^2(\R^2\setminus{\{0\}})$ allows us to define a real structure on $(\g,\BK)$ in the following way. For $T\in \mathbb{D}\subset \operatorname{End}(L^2(\R^2\setminus{\{0\}}))$,
the formula
   \[ \sigma(T)=-T^*,\]
where $T^*$ is the adjoint of $T$, defines a conjugate linear involution
  of    $\mathbb{D}$ which descends to an antiholomorphic involution of $\g$.
Similarly, the action of  $O(2,\C)$ on  $L^2(\R^2\setminus{\{0\}}) $ embeds it into $\operatorname{Aut}(L^2(\R^2\setminus{\{0\}}))$  and we can define an antiholomorphic involution of  $O(2,\C)$ by 
\[\sigma_{O(2,\C)}(g)=(g^*)^{-1}, \]
where here we identify $g\in O(2,\C)$ with its image in $\operatorname{Aut}(L^2(\R^2\setminus{\{0\}}))$. By direct calculation, in terms of the matrix group $O(2,\C)$, $\sigma_{O(2,\C)}$ turns out to be the usual complex conjugation of matrices. 
As was described in Section \ref{sec2.4},  this defines a real structure $\sigma_{\BK}$ on the constant family $\BK$.
\begin{remark}
A standard approach for  solving  the   Schr\"odinger equation $H\psi=E\psi$  is to find a complete  set of  Hermitian operators that commute with $H$. This  was our  motivation for the definition of  $\g$ using the centralizer $C_{\mathbb{D}}(H)$  and for the definition of  the real structure $\sigma$ using the adjoint of a linear operator. 
\end{remark}

 \begin{lemma}
The isomorphism $(\g,\BK)\simeq(\widetilde{\boldsymbol{\mathfrak{so}_3}}, \boldsymbol{O_2}) $ is compatible with the real structure on $(\g,\BK)$, introduced above, and the real structure on $(\widetilde{\boldsymbol{\mathfrak{so}_3}}, \boldsymbol{O_2}) $, introduced in example \ref{ex4}.\qed
\end{lemma}  
By Example \ref{ex4} we know that $\widetilde{\boldsymbol{\mathfrak{so}_3}}$ with its real structure  can be lifted to the family $\widetilde{\boldsymbol{SO_3}}$ with a compatible real structure. We summarize the results of this  section in the following theorem. 
 
 \begin{theorem}
For any $E\in \operatorname{Spec}(H) \subset \X^{\sigma_{\X}}$ the visible symmetry of the  Schr\"{o}dinger equation $H\psi=E\psi$  is  $\BK|_E^{\sigma_{\BK}}$, and the infinitesimal hidden symmetry is   $\g|_E^{\sigma_{\g}}$. Furthermore,   $\g^{\sigma_{\g}}$ can be lifted to a family of Lie groups  that  correspond to the  hidden symmetries. That is,  there is an algebraic family of complex algebraic groups $\BG$ over $\X$ with a real structure $(\sigma_{\X},\sigma_{\BG})$, such that, for every $E\in \X^{\sigma}$, \[\BG|_E^{\sigma_{\BG}}\simeq  \begin{cases}
SO(2,1), & E>0,\\
O(2)\ltimes \R^2, & E=0,\\
SO(3), & E< 0.
\end{cases}\] 
\end{theorem}

\section{The family of Harish-Chandra modules  and hidden symmetries }
In this section  we show how the physical realization of the family $(\g,\BK)$ induces an essentially unique family of  $(\g,\BK)$-modules from which one can recover the solution spaces (as unitary representations) of the Schr\"{o}dinger equation.

\subsection{Algebraic families of Harish-Chandra modules}\label{sec4.1}
Let $(\g,\BK)$ be the  algebraic  family of Harish-Chandra pairs introduced in Section \ref{sec3}. Roughly speaking, an algebraic family of Harish-Chandra modules for $(\g,\BK)$ (or algebraic family of  $(\g,\BK)$-modules) is a family of complex vector spaces parameterized by  $\X$ that carries  compatible actions of $\g$ and $\BK$.  More precisely, an  algebraic family of Harish-Chandra modules is 
a flat quasicoherent $O_{\X}$-module, $\F$, that is  equipped with compatible actions of $\g$ and $\BK$, see \cite[Sec. 2.4]{Ber2016}. 
Since, in our case,  $\X$ is the affine variety $\C$, $\F$ can be identified with its space of global sections and can be considered as  a $\C[E]$-module that carries a representation of $\g$ as a Lie algebra over $\C[E]$ and  a compatible action of $K=O(2,\C)$, the fiber of the constant family $\BK$. We shall freely change our perspective between sheaves of $O_{\X}$-modules and $\C[E]$-modules with no  further warning.  
Since $\BK$  is a constant family,  there is a canonical decomposition of $\F$ into $K$-isotopic subsheaves (or submodules)
\[\F=\oplus_{\tau\in \widehat{K}}\F_{\tau} .\] 
We  call $\tau\in \widehat{K}$ with nonzero   $\F_{\tau} $ a $K$-type of $\F$.  
We shall only consider families  that are quasi-admissible and generically irreducible (see  \cite[Sec. 2.2 $\&$ 4.1]{Ber2016}). For such families,
 each $\F_{\tau}$ is a free  $K$-equivariant $O_{\X}$-module of finite rank. In simpler terms, the space of global sections of each $\F_{\tau}$ is isomorphic to $\C[E]\otimes_{\C}V_{\tau}$,  with $V_{\tau}$ isomorphic to a finite direct sum of copies of the representation $\tau$. The number  of  summands in this direct sum is called the multiplicity of the $K$-type $\tau$ in $\F$. 
In addition, for almost any $x\in \X$ the fiber $\F|_x$ is an irreducible admissible  $(\g|_x,\BK|_x)$-module. For  every $x\in\mathbb{C}^*\subset  \X$, the Harish-Chandra pair  $(\g|_x,\BK|_x)$ is isomorphic to the pair arising from $SO(2,1)$.  That is, the pair $(\mathfrak{so}(2,1)_{\C},O(2,\C))$, where  $\mathfrak{so}(2,1)_{\C}$  is  the complexification of $\mathfrak{so}(2,1)$, and $O(2,\C)$ sits, as before, block diagonally inside  $SL(3,\C)$. 
 The classification of the  irreducible admissible $(\mathfrak{so}(2,1)_{\C},O(2,\C))$  is well known.   For example, it can be deduced  
 from the  classification of the  irreducible admissible $(\mathfrak{sl}_2(\C),SO(2,\C))$-modules given  in \cite[ch.1]{Vogan81}, (or \cite[sec. II. 1]{HoweTan} or \cite[ch. 8]{Taylor86})  combined with Clifford theory \cite{Clifford}. Or it can be directly deduced from \cite{Naimark}.
This determines the possible lists of $K$-types for  generically irreducible and quasi-admissible families of $(\g,\BK)$-modules. Such a list of $K$-types is an invariant of an isomorphism class of  families of $(\g,\BK)$-modules. Using this  and several more invariants, 
the  classification of generically irreducible and quasi-admissible families of Harish-Chandra modules  for a closely related family of Harish-Chandra pairs was given in \cite[Sec. 4]{Ber2016}.
By the same methods one can classify generically irreducible and quasi-admissible families of $(\g,\BK)$-modules. The relation between the family of Harish-Chandra pairs in   \cite{Ber2016} and the family $(\g,\BK)$ introduced in Section \ref{sec3},  is analogous to the relation between $SU(1,1)$ and $SO(2,1)$. More precisely, the families of Lie algebras are isomorphic while the  fibers of the  two constant families of groups are related by quotient by a two element subgroup and extension by a two element group. We note that similar families of representations were studied in \cite{Adams2017}.

 Before we state the needed ingredients  from the mentioned classification, we shall describe the center of the enveloping algebra of $\g$. Let $\mathcal{U}(\g)$ be the universal enveloping algebra of $\g$. By Poincar\'e-Birkhoff-Witt theorem, as a module over $\C[E]$ we can write
\[\mathcal{U}(\g)=\oplus_{i,j,k\in \mathbb{N}_0}  \C[E] {L}^i{A}_x^j{A}_y^k.\]
By direct calculation, $\mathcal{Z}(\g)$, the center of $\mathcal{U}(\g)$, is a free polynomial algebra over $\C[E]$ with one generator  \[\Omega:={A}_x^2+{A}_y^2-E{L}^2\] which we call \textit{the regularized Casimir}.  On a generically irreducible and quasi-admissible family of $(\g,\BK)$-modules, the regularized Casimir  $\Omega$ must act via multiplication by a function $\omega(E)\in \C[E]$, see \cite[Sec. 4.4]{Ber2016}. The function $\omega(E)$ is another  invariant of a generically irreducible and quasi-admissible family of $(\g,\BK)$-modules. The following proposition describes a class of families of $(\g,\BK)$-modules that are completely determined by the above mentioned  invariants. 
\begin{proposition}
Let $\tau$ be  an irreducible algebraic representation of  $O(2,\C)$. Up to an isomorphism, a generically irreducible and quasi-admissible family of $(\g,\BK)$-modules $\F$ that is generated by its $\tau$ isotopic piece $\F_{\tau}$,  is  determined by $\omega(E)$.\label{prop2}
\end{proposition}
The proposition follows from  \cite[Thm. 4.9.3]{Ber2016} and the classification of admissible irreducible $(\mathfrak{sl}_2(\C),SO(2,\C))$-modules. In Section \ref{sec4.3}  we shall be interested in a specific case of such families for which we explicitly show how  $\omega(E)$ determines $\F$.

\subsection{The  families of Harish-Chandra modules imposed by the physical realization}
\label{sec4.2} 
The family $\g$ is more than just an abstract  algebraic family of Lie algebras, it is a family that is given by a concrete realization. The realization is induced by the realization of  the Schr\"{o}dinger operator $H$ as  a differential operator on smooth complex valued functions on $\mathbb{S}=\R^2\setminus \{0\}$.
As such, there are (algebraic) relations between  the elements of  $C_{\mathbb{D}}(H)$. For our purposes,  the relevant  relation is given by 
\begin{eqnarray}\label{Casimir}
A_x^2+A_y^2+\frac{1}{2}k^2=H\left(L^2-\frac{1}{4}\right). 
\end{eqnarray}
This can be verified directly. In fact, this relation is nothing but the Casimir relation for the quantum superintegrable system $E18$ that is determined by $H$.
Now, on (the $SO(2)$-finite functions in) the solution space for the Schr\"{o}dinger equation with eigenvalue $E$, realized  as a space of functions on $\R^2\setminus \{ 0\}$,   the Lie algebra  $\g|_E$ naturally acts via  its realization as differential operators on functions on $\R^2\setminus \{ 0\}$: 
\begin{eqnarray}\nonumber
&& L=y\partial_x-x\partial_y, \\ \nonumber
&& A_y=\frac{i}{\sqrt{2}}\left(y\partial_{xx}-x\partial_{xy}-\frac{1}{2}\partial_{y}+k\frac{y}{r}\right),\\ \nonumber
&& A_x= \frac{i}{\sqrt{2}}\left(x\partial_{yy}-y\partial_{xy}-\frac{1}{2}\partial_{x}+k\frac{x}{r}\right).
\end{eqnarray}
This induces a morphism of algebras  $\mathcal{U}(\g)\longrightarrow C_{\mathbb{D}}(H)$. 
Hence,  (\ref{Casimir}) implies that $\Omega={A}_x^2+{A}_y^2-E{L}^2$ acts via multiplication by the function $\omega(E)=-\frac{E}{4}-\frac{k^2}{2}$. 
 As noted above, one can classify all generically irreducible and quasi-admissible  families of $(\g,\BK)$-modules. These calculations follow from those in \cite[Sec. 4]{Ber2016}. We shall only need the following fact which can be derived from the  classification.
\begin{fact}\label{fact}
Let  $\F$ be a generically irreducible and quasi-admissible  family of $(\g,\BK)$-modules  on which $\Omega$ acts by multiplication by  $\omega(E)=-\frac{E}{4}-\frac{k^2}{2}$.  Then  any irreducible algebraic  representation of $SO(2,\C)$ is an  $SO(2,\C)$-type of $\F$,  and each appears with multiplicity one.     
\end{fact}
A point $E\in \X$ such that $\F|_E$ is reducible is called a \textit{reducibility point} of $\F$. We are now ready to calculate the bounded spectrum of the  Schr\"{o}dinger operator.
\begin{theorem}\label{Theo2}
Let  $\F$ be a generically irreducible and quasi-admissible  family of $(\g,\BK)$-modules  on which $\Omega$ acts via multiplication by  $\omega(E)=-\frac{E}{4}-\frac{k^2}{2}$.   The collection of all the reducibility points of $\F$ coincides with $\mathcal{E}_b$.
\end{theorem}

Before proving Theorem \ref{Theo2} we   recall the natural parameterization of $\widehat{SO(2,\C)}$,  the set of equivalence classes of irreducible algebraic representations of $SO(2,\C)$, and  do the same for  ${O(2,\C)}$.  For every $n\in \Z$, the formula $\chi_n(R(\theta))=e^{in\theta}$   defines a one-dimensional irreducible algebraic representation of $O(2,\C)$, and up to equivalence these are all such representations. This identifies $\Z$ with $\widehat{SO(2,\C)}$.  For $n\in \mathbb{N}$, the $SO(2)$ representation $\chi_n\oplus \chi_{-n}$  can be turned into an irreducible representation of $O(2,\C)$ by defining the action of $s$ on the underlying vector space $\C^2$  by $s\cdot (z,w)=(w,z)$. In addition, the trivial representation of $SO(2,\C)$ can be extended in exactly two ways to an irreducible algebraic representation of $O(2,\C)$. One is the trivial representation of  $O(2,\C)$, denoted by  $\chi_0^+$, and the other is  the  determinant representation of $O(2,\C)$, denoted by $\chi_0^-$. The collection $\{\chi_0^+,\chi_0^-,\chi_n\oplus \chi_{-n}|n\in \mathbb{N}\}$ contains exactly one representative for each class in $\widehat{O(2,\C)}$.

\begin{proof}
Let $\F$ be a family satisfying the  hypotheses of Theorem \ref{Theo2}. By Fact \ref{fact},  there  exists  a sequence   $\{f_n|n\in \Z\} \subset 
 \F$ with   \begin{eqnarray}\nonumber
 && R(\theta)f_{n}=e^{in\theta}f_{n},
\end{eqnarray}  
and such that  the decomposition of $\F$ with respect to the action of $SO(2,\C)$ is given by 
 \[\F=\oplus_{n\in \mathbb{Z}}\F_{n}, \hspace{1cm} \F_n=\C[E]f_n. \]  
 Define a new basis for $\g$ by  $\mathcal{J}=i{L}$, $\mathcal{A}_{+}=\frac{1}{\sqrt{2}}({A}_x+i{A}_y)$, $\mathcal{A}_{-}=\frac{1}{\sqrt{2}}({A}_x-i{A}_y)$. In particular, for every $n\in \Z$, $\mathcal{J}f_n=nf_n$. The commutation relations of these basis elements are given   by 
\begin{eqnarray}
&&[\mathcal{J},\mathcal{A}_{+}]=\mathcal{A}_{+}, [\mathcal{J},\mathcal{A}_{-}]=-\mathcal{A}_{-}, [\mathcal{A}_+,\mathcal{A}_-]=E\mathcal{J}.\label{comm}
\end{eqnarray} 
In terms of the new basis, the regularized Casimir is given by  
\[ \Omega=E\mathcal{J}^2+E\mathcal{J}+2\mathcal{A}_-\mathcal{A}_+=E\mathcal{J}^2-E\mathcal{J}+2\mathcal{A}_+\mathcal{A}_- .\]
Hence, for every $n\in \Z$ we have    
\begin{eqnarray}\nonumber
&& \mathcal{A}_-\mathcal{A}_+f_{n}=\frac{1}{2}\left(\omega(E)-E(n^2+n)\right) f_{n}=-\frac{1}{2}\left(\frac{k^2}{2} +E\left(n+\frac{1}{2}\right)^2\right) f_{n},\\ \nonumber
&& \mathcal{A}_+\mathcal{A}_-f_{n}=\frac{1}{2}\left(\omega(E)-E(n^2-n)\right)f_{n}=-\frac{1}{2}\left(\frac{k^2}{2} +E\left(n-\frac{1}{2}\right)^2\right) f_{n},
\end{eqnarray} 
and  the reducibility points are given by   $\left\{E_n:=-\frac{k^2}{2(n+1/2)^2}\middle|n \in \{0,1,2,...\} \right\}$ which is nothing but $\mathcal{E}_b$.
\end{proof}
\subsection{Concrete  families  of Harish-Chandra modules}\label{sec4.3}
There are many  generically irreducible and quasi-admissible  families of $(\g,\BK)$-modules on which $\Omega$ acts via multiplication by  $\omega(E)=-\frac{E}{4}-\frac{k^2}{2}$.  The analysis that follows can be applied to any  of them.  For concreteness, we shall focus on  such  families that are generated by their isotypic subsheaf corresponding to the trivial $SO(2,\C)$-type. Combining Proposition \ref{prop2} and Fact \ref{fact}, we see that there are at most  two such families. We shall see below that  there are exactly two such families. These two families must have the same $SO(2,\C)$-types but they differ by one $O(2,\C)$-type. One of them contains the trivial $O(2,\C)$-type and the other  contains the determinant $O(2,\C)$-type.  .

From now on, we let  $\F$ be a generically irreducible and quasi-admissible  family of $(\g,\BK)$-modules on which $\Omega$ acts via multiplication by  $\omega(E)=-\frac{E}{4}-\frac{k^2}{2}$  and that  is generated by $\F_0$, the isotypic piece corresponding to the trivial $SO(2,\C)$ representation.  In this case, all the $SO(2,\C)$-types appear, each has multiplicity one, and each $\F_n$ is a free rank-one $\C[E]$-module.  We choose a basis $f_0$ of $\F_0$ and define 
\[f_{n}:=\begin{cases}
(\mathcal{A}_+)^nf_0 & n>0,\\
(\mathcal{A}_-)^{-n}f_0 & n<0.
\end{cases} \]
The commutation relations (\ref{comm})  imply that $f_{n}\in \F_{n}$. Since $\F$ is generated by $\F_0$, all $f_{n}$ vanish nowhere. That is, at any $E\in \C$, each $f_{n}$ defines a nonzero vector in the fiber of $\F$ at $E$.  In particular, $\{f_{n}|n\in \Z\}$ is a basis for $\F$ and $\F_{n}=\C[E]f_{n}$.   This, together with the commutation relations (\ref{comm}), completely determines the action of $\g$ to be given by 
 \begin{eqnarray}\label{4.2}
 && \mathcal{J}f_{n}=nf_{n},\\ \label{4.3}
 && \mathcal{A}_+f_{n}=\begin{cases}
f_{n+1} & n\geq 0,\\
-\frac{1}{2}\left(\frac{k^2}{2} +E\left(n+\frac{1}{2}\right)^2\right)  f_{n+1} & n<0,
\end{cases}\\ \label{4.4}
 && \mathcal{A}_-f_{n}=\begin{cases}
-\frac{1}{2}\left(\frac{k^2}{2} +E\left(n-\frac{1}{2}\right)^2\right) f_{n-1} & n> 0,\\
f_{n-1} & n\leq 0.
\end{cases}
 \end{eqnarray}
 Of course, the action of $SO(2,\C)$ is  given by 
 \[R(\theta)f_{n}=e^{in\theta}f_{n}. \]
 In addition, one can show that  either  one of the two formulas (that only differ by a sign) 
  \begin{eqnarray}
&&s \cdot f_{n}=\pm(-i)^nf_{-n}\label{eq4.8}
\end{eqnarray}  
extends $\F$ to a generically irreducible and quasi-admissible algebraic  family of $(\g,\BK)$-modules over $\X$. The two families obtained  this way are not isomorphic.

\subsection{The Jantzen Filtration}\label{sec4.4}
In this section we recall  the $\sigma$-twisted dual, which is a certain dual family to $\F$. We calculate the space of intertwining operators from $\F$ to its $\sigma$-twisted dual, and we recall how such an intertwining operator gives rise to the Jantzen filtration of the fibers of $\F$.

\subsubsection{The $\sigma$-twisted dual}\label{sec4.4.1}
Using the real structure of $(\g,\BK)$, with any algebraic family $\F$   of $(\g,\BK)$-modules, we can associate another family of $(\g,\BK)$-modules, $\F^{\langle \sigma \rangle}$, \textit{the $\sigma$-twisted dual of $\F$}.  See \cite[Sec. 2.4]{Ber2017} for precise definition. 
For the two families $\F$ that were defined in Section \ref{sec4.3},  the $\sigma$-twisted dual has a basis $\{q_{n}|n\in \Z\}$ such that 
\[\F^{\langle \sigma \rangle}=\oplus_{n\in \mathbb{Z}}\F^{\langle \sigma \rangle}_{n}\hspace{2mm}\text{with}\hspace{4mm} \F^{\langle \sigma \rangle}_{n}=\C[E]q_{n}, \] 
and the action of  $(\g,\BK)$ is determined by 
\begin{eqnarray}\label{4.20}
 && \mathcal{J}q_{n}=nq_{n},\\ \label{4.30}
 && \mathcal{A}_+q_{n}= \begin{cases}
\frac{1}{2}\left(\frac{k^2}{2}+E\left(n+\frac{1}{2}\right)^2\right)q_{n+1} & n\geq 0,\\
-q_{n+1} & n< 0,
\end{cases}   \\ \label{4.40}
 && \mathcal{A}_-q_{n}=\begin{cases}
-q_{n-1} & n> 0,\\
\frac{1}{2}\left(\frac{k^2}{2}+E\left(n-\frac{1}{2}\right)^2\right)q_{n-1} & n\leq 0,
\end{cases}\\
 &&R(\theta)q_{n}=e^{in\theta}q_{n}, \hspace{8mm}s \cdot q_{n}=\pm(-i)^nq_{-n}. \label{eq4.12}
 \end{eqnarray}  
   \subsubsection{Intertwining operators } 
 In this section we describe the space of intertwining operators from $\F$ to $\F^{\langle \sigma \rangle}$.
 \begin{proposition}\label{propo3}
$\operatorname{Hom}_{(\g,\BK)}(\F,\F^{\langle \sigma \rangle})\simeq \C[E]$.
 \end{proposition}
  
\begin{proof}
By definition, $\operatorname{Hom}_{(\g,\BK)}(\F,\F^{\langle \sigma \rangle})$ is the space of algebraic intertwining operators from $\F$ to $\F^{\langle \sigma \rangle}$. Any such operator is given by a morphism of $\C[E]$-modules  $\psi:\F\longrightarrow \F^{\langle \sigma \rangle}$  that is equivariant  with respect to the actions of $\g$ and $\BK$. The equivariance with respect to $\BK$ implies that for every $n\in \Z$, 
 \[\psi(f_{n})=\psi_{n}(E)q_{n} \]  for some $\psi_{n}\in \C[E]$, with   \[\psi_{-n}(E)=\psi_{n}(E). \] 
  The equivariance with respect to $\g$ further implies that  \[\psi( \mathcal{A}_{\pm}f_{n})= \mathcal{A}_{\pm}\psi_{n}(E)q_{n}, \hspace{1mm} \forall n\in \Z. \] 
 This leads to a recursion relation between the $\psi_n(E)$ and the explicit form of $\psi_n(E)$ is  given by 
 \begin{eqnarray}\label{eq411}
&& \psi_{n}(E)=\psi_{-n}(E)=\frac{1}{2^{|n|}}\psi_{0}(E)\Pi_{m=1}^{|n|}\left(\frac{k^2}{2}+E\left(m-\frac{1}{2}\right)^2\right).
\end{eqnarray} 
 Hence, the map that assigns to $\psi \in \operatorname{Hom}_{(\g,\BK)}(\F,\F^{\langle \sigma \rangle})$  the function $\psi_0(E)\in \C[E]$  is an isomorphism of $\C[E]$-modules. 
\end{proof}  
\subsubsection{The Jantzen filtration}
In this section we describe the Jantzen filtration. For more  information see \cite[Sec. 4.1]{Ber2017}. 
Let $\psi:\F\longrightarrow \F^{\langle \sigma \rangle}$  be a nonzero intertwining operator.  
For  a fixed $e \in \C$,  we can localize $\C[E]$ at $e$, obtaining the ring $\C[E]_e$ of all rational functions that are defined near $e$. Explicitly,  \[\C[E]_e=\left\{\frac{f}{g}| f,g\in \C[E], g(e)\neq 0\right\}. \] Similarly, we can form the localizations of $\F$ and  $\F^{\langle \sigma \rangle}$ at $e$, \[\F_e:=\C[E]_e\otimes_{\C[E]}\F,\hspace{5mm}  \F^{\langle \sigma \rangle}_e:=C[E]_e\otimes_{\C[E]}\F^{\langle \sigma \rangle}. \]  
The morphism $\psi$ induces a morphism of $\C[E]_e$-modules from $\F_e$ to  $ \F^{\langle \sigma \rangle}_e$ that, by abuse of notation, we shall also denote by $\psi$. 
Using $\psi$, we obtain a  decreasing filtration of  $\F_e$
defined by  
\begin{eqnarray}\nonumber
&& \F_e^n=\{f\in \F_e|\psi(f)\in (E-e)^n\F_e\},
\end{eqnarray}
for every $n\in \mathbb{N}_0=\{0,1,2,...\}$.
The fiber of $\F$  at $e$ is defined by
 \[\F|_e:=\C\otimes_{\C[E]}\F=\C\otimes_{ \C[E]_e}\F_e.\]
The natural surjection from the localization at $e$ to the fiber at $e$ gives rise to a decreasing  filtration $\{\F|_e^n\}$ of $\F|_e$. This  is the \textit{Jantzen filtration}.   The Jantzen quotients (of $\F|_e$) are $\F|^n_e/\F|^{n-1}_e$. Up to a shift of the filtration parameter $n$, the Jantzen quotients are independent of (a nonzero) $\psi$. We say that the filtration is trivial if $0 \neq \F|_e^n$ implies $\F|_e^n=\F|_e$.

  \begin{proposition}
Let  $\F$ and $\psi$  be as above and let $e\in \C$. The Jantzen filtration of $\F|_e$ is trivial   iff $e\in \C\setminus{\mathcal{E}_b}$. If $e\in \mathcal{E}_b$ there are exactly two nonzero quotients. Furthermore, for $e=e_m=-\frac{k^2}{2(m+1/2)^2}$  with $m\in \{0,1,2,...\}$,   one of the quotients is of dimension $2m+1$  with  $SO(2,\C)$-types $\{-m,-m+1,...,m \}$  and the other is infinite dimensional with  $SO(2,\C)$-types $\{l\in \Z| |l|>m \}$\hspace{1pt}.
\label{propos4}
\end{proposition}
\begin{proof}
As mentioned above, the quotients are independent of $\psi$, so we may and will  assume that $\psi_0(E)$ is the constant function $1$.   Note that by  (\ref{eq411}), $\psi(f_{m})=\psi_{m}(E)q_{n}$, where $\psi_m(E)$ is a polynomial of degree $m$   and its $m$ different roots are \[\left\{e_n=-\frac{k^2}{2(n+1/2)^2}\middle| 0\leq n< m, n\in \Z \right\}\subset \mathcal{E}_b.\] 
Hence, for any $e\in \C\setminus{\mathcal{E}_b}$, none of the morphisms $\psi_n(E)$ vanishes at $e$,  and therefore $\F|_e^0=\F|_e$, $\F|_e^1=\{0\}$, and the unique nonzero quotient is $\F|^1_e/\F|^{0}_e\simeq  \F|_e$. For $e=e_m=-\frac{k^2}{2(m+1/2)^2}$,   all the morphisms $\psi_n(E)$ with $|n|\leq m$ do not vanish  at  $e_m$ and all  other  $\psi_n(E)$  do. Hence, as $SO(2,\C)$-modules, the nonzero quotients   are 
\[\F|^0_{e_m}/\F|^{1}_{e_m}\simeq  \chi_{-m}\oplus  \chi_{-m+1}\oplus .. .\oplus  \chi_{m},\]
\[\F|^1_{e_m}/\F|^{2}_{e_m}\simeq  \bigoplus_{|n|\geq m} \chi_{n}.\]
\end{proof}
\begin{remark}
The $O(2,\C)$-types of the Jantzen quotients come in two flavors, corresponding to the two families $\F$ with  $\omega(E)=-\frac{E}{4}-\frac{k^2}{2}$. In one of the families $\F$,  all the trivial $SO(2,\C)$ representations  in the various quotients extend to  the  trivial $O(2,\C)$ representation. In the other,  all the trivial $SO(2,\C)$ representations  in the various quotients extend to the nontrivial one-dimensional $O(2,\C)$ representation.
\end{remark}

\subsection{The Hermitian form on the Jantzen quotients}
In this section we describe the invariant Hermitian form  on the Jantzen quotients. For more  information see \cite[Sec. 4.2]{Ber2017}.

The module 
$\F^{\langle \sigma \rangle}$ can be naturally identified with the space of functions from $\F$ to $\C[E]$ that  are conjugate $\C[E]$-linear and vanish  on all but finitely many of the $K$ isotypic summands of $\F$.
That is, for $h\in \F^{\langle \sigma \rangle}$,  $f\in \F$, and $p\in \C[E]$, 
\[ h(pf)=\sigma(p)h(f),\]   where  $\sigma(p)(E)=\overline{p(\overline{E})}$. The bases $\{f_n\}$ and $\{q_n\}$  can be chosen such that under the abovementioned identification,  
\[ q_n(f_m)=\delta_{mn}.\]
 This allows us to equip each of the  Jantzen  quotients, $\F|^n_e/\F|^{n-1}_e$, with  a non-degenerate sesquilinear form $\langle \_,\_ \rangle_{e,n}$ defined by 
\[\langle  [f],[f'] \rangle_{e,n}=\left((E-e)^{-n}\psi(f)(f') \right)_{E=e}.\]
That is, for every $[f],[f']\in \F|_e^n$ with representatives $f,f'\in \F$, we apply the linear functional $\psi(f)$ to $f'$, then  we divide by $(E-e)^{n}$ (this makes sense since $[f],[f']\in \F|_e^n$) and finally  we evaluate this complex valued  function of $E$ at $e$.  The forms $\langle \_,\_ \rangle_{e,n}$  are invariant under the actions of $\BK|_e$ and $\g|_e$ in the the  sense that
\begin{eqnarray}\nonumber
&&
\langle X  \cdot [f] , [f'] \rangle_{e,n}  + \langle [f],  \sigma(X)\cdot [f'] \rangle_{e,n} = 0, \\ \nonumber
&&\langle g  \cdot [f] , \sigma_{\boldsymbol{K}} (g)  \cdot [f'] \rangle_{e,n}  =\langle [f],  [f'] \rangle_{e,n},  
\end{eqnarray}
for all $[f],[f']\in  \F|_e^n$,   $X\in \g|_e$, and  $g\in \BK|_e$. Up to a shift of the filtration parameter $n$, the Jantzen quotients are independent of (a nonzero) $\psi$. The intertwining operator $\psi$ can be chosen such that for every $e\in \R\subset \C$, the  forms $\langle \_,\_ \rangle_{e,n}$ will be Hermitian. If the invariant Hermitian form is  of a definite sign, we say that  the Jantzen quotient is \textit{infinitesimally unitary}. 
\begin{proposition}\label{propos5}
Let  $\F$ and $\psi$   be as above with $\psi_0(E)\equiv 1$. The form is of a definite sign (positive definite or negative definite) exactly on the following Jantzen quotients:
\begin{enumerate}
\item $\F|^0_e/\F|^{1}_e\simeq  \F|_e$ with $e>0$. 
\item $\F|^0_0/\F|^{1}_0\simeq  \F|_0$.
 \item $\F|^0_{e_m}/\F|^{1}_{e_m}$ with $e_m=-\frac{k^2}{2(m+1/2)^2}$, for some $m\in \mathbb{N}_0$.
\end{enumerate}
\end{proposition}
\begin{proof}
By direct calculation on $\F|^n_e/\F|^{n+1}_e$,  the form satisfies 
\[\langle  [f_s],[f_t] \rangle_{e,n}=\left((E-e)^{-n}\frac{1}{2^{|s|}}\Pi_{l=1}^{|s|}\left(\frac{k^2}{2}+E\left(l-\frac{1}{2}\right)^2\right) \right)_{E=e}\delta_{st}.\]
If $e$ is not real, then the form can not be of definite sign. If 
$e\in \R\setminus{\mathcal{E}_b}$, then by  Proposition \ref{propos4},  the unique nonzero Jantzen quotient is $\F|^0_e/\F|^{1}_e$ on which the form is given by 
\[\langle  [f_s],[f_t] \rangle_{e,0}=\left(\frac{1}{2^{|s|}}\Pi_{l=1}^{|s|}\left(\frac{k^2}{2}+e\left(l-\frac{1}{2}\right)^2\right) \right)\delta_{st},\] which is obviously positive definite for $e\geq 0$, and not of  definite sign for $e<0$. The case with $e=e_m$ is proven similarly. 
\end{proof}
For each $e\in \mathcal{E}$ we shall denote by $J(\F|_e)$ the unique infinitesimally unitary Jantzen quotient of $\F|_e$.   Combining  Proposition \ref{propos4}
and Proposition \ref{propos5}, we obtain the following. 

\begin{theorem}\label{th3}
Let  $\F$ be any one of the two generically irreducible and quasi-admissible  families of $(\g,\BK)$-modules on which $\Omega$ acts via multiplication by  $\omega(E)=-\frac{E}{4}-\frac{k^2}{2}$,  and that are generated by $\F_0$, their isotypic piece corresponding to the trivial $SO(2,\C)$ representation. The spectrum of $H$ coincides with the set of all  $E\in \X$ for which $\F|_E$  has  a nonzero infinitesimally unitary Jantzen quotient. Moreover,  $\mathcal{E}_b$   coincides with the set of all  $E\in \X$ for which $\F|_E$ has a nontrivial Jantzen filtration.
\end{theorem}

\subsection{Jantzen quotients, group representations, and solution spaces}\label{apen}

In this section, after recalling the physical solution spaces for the Schr\"{o}dinger equation,  we describe how the infinitesimally unitary Jantzen quotients  can be integrated to  unitary group representations. We show that these group representations are isomorphic to the solution spaces of  the Schr\"{o}dinger equation, as unitary represenations.

\subsubsection{The Physical Solutions}
The physical solutions for the Schr\"{o}dinger equation $H\psi=E\psi$ are  complex valued functions on $\R^2$ which solve the equation, which are regular  at the origin and have the best possible behavior at  infinity e.g., see \cite[Chap. V.]{MR0093319}. For every  $ E\in \operatorname{spec}(H)$, the collection of physical solutions corresponding to $E$  can be completed to a Hilbert space, $Sol(E)$, that carries a unitary irreducible representation (of the corresponding symmetry group). 
This  nontrivial fact was first discovered by Fock, for bound states in the three dimensional case \cite{Fock:1935vv}. For the two-dimensional case  see, e.g.,  \cite{Bander1,Bander2} and \cite{Torres1998}.
We remark that for any $ E\in \operatorname{spec}(H)$,  $Sol(E)$ is a tempered representation, and for $E\in \mathcal{E}_b$ it is  a square-integrable representation. 
 The space  of $SO(2)$-finite vectors of  $Sol(E)$, denoted $Sol(E)_{SO(2)-\text{f.}}$,  coincides with the space of all physical solutions that are separable in polar coordinates of $\R^2$;  $\psi(r,\varphi)=R(r)\Phi(\varphi)$.
We shall explicitly specify a basis for $Sol(E)_{SO(2)-\text{f.}}$. For nonzero energy  we shall follow  \cite{PhysRevA.43.1186}. For zero energy we shall follow \cite{Torres1998}.

For $0\neq E\in \operatorname{spec}(H)$, the Schr\"{o}dinger equation can be solved via separation of variables in polar coordinates. Any separable solution $\psi(r,\varphi)=R(r)\Phi(\varphi)$  has $\Phi(\varphi)$ proportional to $e^{il\varphi}$, for some $l\in \Z$. For fixed such  $l$, the radial equation  has a two-dimensional solution space. Imposing the condition of regularity at the origin, we are left with a one-dimensional solution space. Explicitly, for a fixed $E_n=-\frac{k^2}{2(n-1/2)^2} 
<0$  with  $n\in \{1,2,3,...\}$, the radial solutions are spanned by 
\[(\beta_nr)^{|l|}e^{-\beta_n r/2}   \ _1F_1(-n+|l|+1;2|l|+1;\beta_nr) \]
where $\beta_n=\frac{2 k}{(n-1/2)}$, $| l |$ varies in $\{0,1,2,...,n-1\}$, and $ \ _1F_1(\alpha:\gamma;z)$ is the confluent hypergeometric function (which reduces to a polynomial function in $r$ for the above values of parameters).  For a fixed $E>0$, the radial solutions are spanned by
\[(2\sqrt{2E} r)^{|l|}e^{-i\sqrt{2E} r} \ _1F_1(-i\sqrt{\frac{k^2}{2E}}+|l|+\frac{1}{2};2|l|+1;i2\sqrt{2E} r)\]
with $| l |$ varying over all non-negative integers. For $E=0$, upon applying the Fourier transform and a linear  change of variables, the   Schr\"{o}dinger equation takes an integral form with solutions that are separable in polar coordinates on $\R^2$. As before, the solutions are of the form $\psi(r,\varphi)=R(r)\Phi(\varphi)$  with  $\Phi(\varphi)$ proportional to $e^{il\varphi}$ for some $l\in \Z$. As in  the case with $E\neq0$, for  fixed  $l\in \Z$, the corresponding solution space for the radial function is two-dimensional. Demanding regularity at $r=0$, the   solution space is one dimensional. Explicitly,  the solution space is spanned by  
\[e^{il\varphi}J_{l}(r), \]
with $J_{l}(r)$ being the Bessel function of the first kind. 

\subsubsection{From Jantzen quotients to the solution spaces}

In Proposition \ref{propos5}, we saw that  each $e\in \mathcal{E}\subset \R$ has  exactly one infinitesimally unitary Jantzen quotient, $J(\F|_e)$. Recall that $J(\F|_e)$   is a  $(\g|_e,\BK|_e)$-module.  In fact,  each quotient is an irreducible   $(\g|_e,\BK|_e)$-module. The discussion in \cite[Sec. 4.3]{Ber2017} explains how a theorem of Nelson \cite{Nelson59} implies that,  as a representation of $\g|_e^{\sigma}$, for any $e\in \mathcal{E}$,  $J(\F|_e)$ can be integrated to a unitary  representation  of the simply connected Lie group with Lie algebra $\g|_e^{\sigma}$. Here we shall proceed using a different approach. We shall   separately deal with   each of the cases   $e>0$, $e=0$, and $e=e_m$  for some $m\in \mathbb{N}_0$, and  explain  why the representations  of the Lie algebras  on  the various $J(\F|_e)$  can be integrated to unitary irreducible representations of $SO_0(2,1)$,  $SO(2)\ltimes \mathbb{R}^2$ and $SO(3)$, depending on the value of $e$. \\   

\textbf{The  $SO_0(2,1)$ case }\\

For $e>0$,  $SO(2,1)\simeq \widetilde{\boldsymbol{SO_3}}|_e^{\sigma}$ is reductive with  maximal compact subgroup $O(2)\simeq \BK|_e^{\sigma}$. By a theorem  of Harish-Chandra and Lepowsky  (see e.g., \cite[Sec. 0.3]{Vogan81}), the Jantzen quotient $J(\F|_e)$  can be  integrated to a unitary irreducible representation of $SO(2,1)$.  The $O(2)$-types and infinitesimal character are enough to parametrize the unitary dual of  $SO(2,1)$.   Using this we can determine that for each of the two families $\F$,  the corresponding unitary irreducible representation  of    $SO(2,1)\simeq \widetilde{\boldsymbol{SO_3}}|_e^{\sigma}$ obtained by integrating $J(\F|_e)$  is a unitary principal series on which the Casimir acts by multiplication by  $\omega(e)$. Both possibilities restrict to the same  unitary irreducible representation of $SO_0(2,1)$, the connected component of $SO(2,1)$.  This representation is isomorphic to the solution space of the   Schr\"{o}dinger equation of the hydrogen atom in two dimensions with energy eigenvalue $e$, compare with  \cite{Bander2}.\\   

\textbf{The $SO(2)\ltimes\mathbb{\R}^2$ case  }\\

The unitary duals of $O(2)\ltimes \mathbb{\R}^2$ and $SO(2)\ltimes\mathbb{\R}^2$ can be  calculated using the \textit{Mackey machine}. Again,  the $O(2)$-types in the case of $O(2)\ltimes \mathbb{\R}^2$ and the $SO(2)$-types in the case of $SO(2)\ltimes \mathbb{\R}^2$ together with  the action of the center of the enveloping algebra are sufficient  to parameterize the unitary duals.  For each of the two families $\F$, it is straightforward, in this low-dimensional case, to find the  unitary irreducible representation of  $O(2)\ltimes \mathbb{R}^2 \simeq \widetilde{\boldsymbol{SO_3}}|_0^{\sigma}$  that have their underlying  Harish-Chandra module given by $J(\F|_0)$.  As in the case of $e>0$, the two non-isomorphic irreducible unitary representations of $O(2)\ltimes \mathbb{R}^2$ have the same unitary irreducible restriction to $SO(2)\ltimes \mathbb{R}^2$. This representation is  isomorphic to the solution space of the   Schr\"{o}dinger equation of the hydrogen atom in two dimensions with energy eigenvalue $e=0$, compare with  \cite{Torres1998}.\\   

\textbf{The $SO(3)$ case  }\\

For $e=e_m$ with  $m\in \mathbb{N}_0$, $J(\F|_{e_m})$ is a $(2m+1)$-dimensional irreducible representation  of $\mathfrak{so}(3)\simeq \g|_{e_m}$ and hence can be integrated to the unique unitary irreducible   $(2m+1)$-dimensional representation of $SO(3)\simeq \widetilde{\boldsymbol{SO_3}}|_{e_m}^{\sigma}$. This representation is  isomorphic to the solution space of the   Schr\"{o}dinger equation of the hydrogen atom in two dimensions with energy eigenvalue $e=e_m=-\frac{k^2}{2(m+1/2)^2}$, compare with   \cite{Bander1}.

Summarizing, we have shown the following result.

\begin{theorem}\label{th4}
Let  $\F$ be as in Theorem \ref{th3}. Then for any $e\in  \operatorname{Spec}(H)$, the Jantzen quotient $J(\F|_e)$ can be integrated to the unitary irreducible representation of the connected component of $\widetilde{\boldsymbol{SO_3}}|_{e}^{\sigma}$ which is isomorphic to $Sol(e)$. 
\end{theorem}

 \section{Discussion}\label{d} 
 \subsection{Physical considerations}

The energy spectrum of the bound states of the hydrogen atom is  a fundamental  example for quantization. In this paper we have shown how quantization of the spectrum  is related to  discretely appearing, infinitesimally unitary Jantzen quotients. We believe that  this connection should be studied further. It seems that similar results hold for the $n$-dimensional hydrogen atom, and we shall treat this case in the future.   But already in two dimensions an interesting feature is the algebraic/analytic structure of the symmetries and  solutions of the hydrogen atom. In principle, knowing $\F$ over  some Zariski open subset of $\X=\C$ is enough to completely  determine $\F$ everywhere on $\C$. For example, this implies that if one knows the  solution spaces for the Schr\"{o}dinger equation with energy eigenvalues confined to some open sub-interval of $(0,\infty)$, one can recover all solutions for all possible energy values. This phenomenon might  be useful in experiments.  One may wonder what should be the physical meaning of this; do the solution spaces for scattering states know about the   bound state solution spaces and vise versa? As we have shown, in the two-dimensional case the answer is affirmative. 

From a  physics perspective there are a couple of natural questions that need to be asked. 
\begin{itemize}
\item Is there a physical meaning that can be attached to the fibers $\F|_e$ with $e$ not in  the spectrum of $H$?
\item Is there a physical explanation for the existence of the  two (rather than just one) families of $(\g,\BK)$-modules that lead to the  solutions of the  Schr\"{o}dinger equation? \end{itemize}

\subsection{Spectral theory considerations}\label{measure}
Up until now, we have solely considered algebraic aspects of the 2D Schr\"{o}dinger equation of the hydrogen atom. But there are various  analytic aspects  that were studied extensively in the past. Most notable among them  is the spectral decomposition for  the Schr\"{o}dinger operator. We shall use the last part of this section to  briefly review some  aspects of  spectral theory. We then speculate on possible connections  between  these analytic  aspects and the algebraic families  introduced  in the paper.  We intend to further study these relations elsewhere.

For $n\geq 3$, the  Schr\"{o}dinger operator, $H=-\frac{\hbar^2}{2\mu}\triangle-\frac{k}{r}$, with domain $C_c^{\infty}(\R^n)\subset L^2(\R^n)$, is essentially self-adjoint. That is, it  has a unique self-adjoint extension. On the other hand,  for $n=2$, and with domain  $C_c^{\infty}(\R^2\setminus{\{0\}})$,   the Schr\"{o}dinger operator has a family of self-adjoint extensions parameterized by $S^1$, see  \cite{DEOLIVEIRA2009251} (The Schr\"{o}dinger operator is not well-defined as an unbounded Hilbert space operator on the domain $C_c^{\infty}(\R^2)$). However, $H$ is semi-bounded \cite{MR0133586} and therefore it has a distinguished  Friedrichs  extension. For this extension, one can use the spectral theory for self-adjoint operators e.g., \cite[Ch. 8]{MR1395149} or \cite[Ch. I Sec. 4]{MR3467631} and obtain a direct-integral decomposition of 
   $L^2(\R^2)$    such that 
\begin{eqnarray}\nonumber
&&L^2(\R^2)\cong \int^{\oplus}_{\operatorname{spec}(H)}\mathcal{H}_{\lambda}d\mu(\lambda), \\ \nonumber
&&(Hs)(\lambda)=\lambda s(\lambda),
\end{eqnarray}
for every section $s$ of $\int^{\oplus}_{\operatorname{spec}(H)}\mathcal{H}_{\lambda}d\mu(\lambda)$ that lies in the domain of $H$. 
For more information on direct integrals see, e.g.,  \cite[Chapter XII]{dixmier1982c},   \cite[Chapter I]{dixmier1981neumann}.
The inclusion of the dense subspace  $S:=C_c^{\infty}(\R^2\setminus{\{0\}})$ into $L^2(\R^2)$  factors through a  Hilbert-Schmidt transformation. 
Hence we can use the    Gelfand-Kostyuchenko method \cite[Sec. 1]{MR1075727}, \cite[Chap. XVII]{maurin1967methods}, \cite[Chap. I]{MR3467631}, and for almost every $\lambda\in  \operatorname{spec}(H)$ obtain  an injective linear map 
\[r_{\lambda}:\mathcal{H}_\lambda \longrightarrow S^+,\] 
where $S^+$  is the Hermitian dual of $S$, which we can view as the space of   distributions on $\R^2\setminus{\{0\}}$.
 Moreover,  $r_{\lambda}(\mathcal{H}_\lambda)\subset   S^+_{\lambda}$ where $S^+_{\lambda}$ is the eigenspace of $H$ with eigenvalue $\lambda$ in the distribution space $S^+$.   
Hence, we obtain an injective map, for which we keep the same notation:
\[{r}_{\lambda}:\mathcal{H}_\lambda \longrightarrow S^+_\lambda .\] 
Elliptic regularity for $H$ implies that the space of eigendistributions, $S^+_\lambda $, is canonically isomorphic to $C_{\lambda}^{\infty}(\R^2\setminus \{ 0\})$, the eigenspace of $H$ with eigenvalue $\lambda$ in $C^{\infty}(\R^2\setminus \{ 0\}) $. Hence, for almost every $\lambda$ we have an injection  
\[\widetilde{r}_{\lambda}:\mathcal{H}_\lambda \longrightarrow C_{\lambda}^{\infty}(\R^2\setminus \{ 0\}) .\]
The space $C^{\infty}(\R^2\setminus \{ 0\})$ carries an action of $C_{\mathbb{D}}(H)$, the centralizer of $H$ (see Section \ref{cent}), and hence also of the family of Lie algebras $\g$. It also carries a compatible action of $K_0:=SO(2)$. This implies that  the space $(C^{\infty}(\R^2\setminus \{ 0\}))_{K_0-\text{f.}}$ of $K_0$-finite  vectors in $C^{\infty}(\R^2\setminus \{ 0\})$ is a $(\g,K_0)$-module.  By explicit calculation, it can be shown that  $(C^{\infty}_{\lambda}(\R^2\setminus \{ 0\}))_{K_0-\text{f.}}$ is a quasi-admissible  $(\g|_{\lambda},K_0)$-module. In fact, the multiplicity of each irreducible representation of $K_0$   in $(C^{\infty}_{\lambda}(\R^2\setminus \{ 0\}))_{K_0-\text{f.}}$ is at most  two. 
Since $SO(2)$ acts on   $L^2(\R^2)$  via a one parameter group of unitary operators that commute with $H$,   almost every  $\mathcal{H}_\lambda$ is equipped with a unitary action of $SO(2)$. The maps $\widetilde{r}_{\lambda}$ are equivariant with respect to the actions of $SO(2)$ and injectivity implies that  $(\mathcal{H}_\lambda)_{K_0-\text{f.}}$ is a quasi-admissible representation of $SO(2)$. The image of $(\mathcal{H}_\lambda)_{K_0-\text{f.}}$ under $\widetilde{r}_{\lambda}$ is a sub-$(\g|_{\lambda},K_0)$-module, for almost every $\lambda$. We are now ready to state a conjecture that relates our algebraic approach to spectral theory.
\begin{conjecture}
For almost every $\lambda\in  \operatorname{spec}(H)$, 
\[\widetilde{r}_{\lambda}((\mathcal{H}_\lambda)_{K_0-\text{f.}})=Sol(\lambda)_{K_0-\text{f.}}.\]
\end{conjecture} 
Note that by Theorem \ref{th4}, the  $(\g|_{\lambda},K)$-module $Sol(\lambda)_{K_0-\text{f.}}$ is isomorphic to  $J(\F|_{\lambda})$. So if the conjecture holds,  for almost every $\lambda\in  \operatorname{spec}(H)$, the abstract space $(\mathcal{H}_\lambda)_{K_0-\text{f.}}$ is a $(\g|_{\lambda},K_0)$-module isomorphic to  $J(\F|_{\lambda})$.

\bibliography{references}

\newcommand{\etalchar}[1]{$^{#1}$}
\def\cprime{$'$}
\begin{thebibliography}{KKPM01}

\bibitem[Ada94]{Adams}
Barry~G. Adams.
\newblock {\em Algebraic approach to simple quantum systems}.
\newblock Springer-Verlag, Berlin, 1994.

\bibitem[Ada17]{Adams2017}
J.~Adams.
\newblock Deforming representations of {SL(2,R)}.
\newblock Preprint, 2017.
\newblock \href{https://arxiv.org/abs/1701.05879}{arXiv:1701.05879}.

\bibitem[Bar36]{Bargmann}
V.~Bargmann.
\newblock {Zur Theorie des Wasserstoffatoms}.
\newblock {\em Z. Phys.}, 99:576--582, 1936.

\bibitem[BBG97]{Bernstein97}
Joseph Bernstein, Alexander Braverman, and Dennis Gaitsgory.
\newblock The {C}ohen-{M}acaulay property of the category of
  {$(\mathfrak{g},K)$}-modules.
\newblock {\em Selecta Math. (N.S.)}, 3(3):303--314, 1997.

\bibitem[Ber88]{MR1075727}
Joseph~N. Bernstein.
\newblock On the support of {P}lancherel measure.
\newblock {\em J. Geom. Phys.}, 5(4):663--710 (1989), 1988.

\bibitem[BHS18a]{Bar2017}
Dan Barbasch, Nigel Higson, and Eyal Subag.
\newblock Algebraic families of groups and commuting involutions.
\newblock {\em Internat. J. Math.}, 29(4):1850030, 18, 2018.

\bibitem[BHS18b]{Ber2016}
J.~Bernstein, N.~Higson, and E.~M. Subag.
\newblock Algebraic families of \uppercase{H}arish-\uppercase{C}handra pairs.
\newblock {\em IMRN}, rny147, 2018.

\bibitem[BHS18c]{Ber2017}
Joseph Bernstein, Nigel Higson, and Eyal~M. Subag.
\newblock Contractions of representations and algebraic families of
  {H}arish-{C}handra modules.
\newblock {\em IMRN}, rny146, 2018.

\bibitem[BI66a]{Bander1}
M.~Bander and C.~Itzykson.
\newblock Group theory and the hydrogen atom. {I}.
\newblock {\em Rev. Modern Phys.}, 38:330--345, 1966.

\bibitem[BI66b]{Bander2}
M.~Bander and C.~Itzykson.
\newblock Group theory and the hydrogen atom. {II}.
\newblock {\em Rev. Modern Phys.}, 38:346--358, 1966.

\bibitem[Cli37]{Clifford}
A.~H. Clifford.
\newblock Representations induced in an invariant subgroup.
\newblock {\em Ann. of Math. (2)}, 38(3):533--550, 1937.

\bibitem[CM69]{Cisneros69}
Arturo Cisneros and Harold~V. McIntosh.
\newblock Symmetry of the two-dimensional hydrogen atom.
\newblock {\em Journal of Mathematical Physics}, 10(2):277--286, 1969.

\bibitem[Dix81]{dixmier1981neumann}
J.~Dixmier.
\newblock {\em Von Neumann algebras}.
\newblock North-Holland mathematical library. North-Holland Pub. Co., 1981.

\bibitem[Dix82]{dixmier1982c}
J.~Dixmier.
\newblock {\em C*-algebras}.
\newblock North-Holland mathematical library. North-Holland, 1982.

\bibitem[dOV09]{DEOLIVEIRA2009251}
César~R. de~Oliveira and Alessandra~A. Verri.
\newblock Self-adjoint extensions of coulomb systems in 1, 2 and 3 dimensions.
\newblock {\em Annals of Physics}, 324(2):251 -- 266, 2009.

\bibitem[DR85]{Dooley85}
A.~H. Dooley and J.~W. Rice.
\newblock On contractions of semisimple {L}ie groups.
\newblock {\em Trans. Amer. Math. Soc.}, 289(1):185--202, 1985.

\bibitem[Foc35]{Fock:1935vv}
V.~Fock.
\newblock {Zur Theorie des Wasserstoffatoms}.
\newblock {\em Z. Phys.}, 98:145--154, 1935.

\bibitem[Gil94]{Gilmore05}
R.~Gilmore.
\newblock {\em Lie groups, {L}ie algebras, and some of their applications}.
\newblock Robert E. Krieger Publishing Co., Inc., Malabar, FL, 1994.
\newblock Reprint of the 1974 original.

\bibitem[GS90]{GuilleminandSternberg}
Victor Guillemin and Shlomo Sternberg.
\newblock {\em Variations on a theme by {K}epler}, volume~42 of {\em American
  Mathematical Society Colloquium Publications}.
\newblock American Mathematical Society, Providence, RI, 1990.

\bibitem[GV16]{MR3467631}
I.~M. Gel'fand and N.~Ya. Vilenkin.
\newblock {\em Generalized functions. {V}ol. 4}.
\newblock AMS Chelsea Publishing, Providence, RI, 2016.
\newblock Applications of harmonic analysis, Translated from the 1961 Russian
  original [ MR0146653] by Amiel Feinstein, Reprint of the 1964 English
  translation [ MR0173945].

\bibitem[HT92]{HoweTan}
Roger Howe and Eng-Chye Tan.
\newblock {\em Nonabelian harmonic analysis}.
\newblock Universitext. Springer-Verlag, New York, 1992.
\newblock Applications of ${{\rm{S}}L}(2,{{\bf{R}}})$.

\bibitem[IW53]{Inonu-Wigner53}
E.~Inonu and E.~P. Wigner.
\newblock On the contraction of groups and their representations.
\newblock {\em Proc. Nat. Acad. Sci. U. S. A.}, 39:510--524, 1953.

\bibitem[KKPM01]{Miller01}
E.~G. Kalnins, J.~M. Kress, G.~S. Pogosyan, and W.~Miller, Jr.
\newblock Completeness of superintegrability in two-dimensional
  constant-curvature spaces.
\newblock {\em J. Phys. A}, 34(22):4705--4720, 2001.

\bibitem[KMW76]{Miller76}
E.~G. Kalnins, W.~Miller, Jr., and P.~Winternitz.
\newblock The group {${\rm O}(4)$}, separation of variables and the hydrogen
  atom.
\newblock {\em SIAM J. Appl. Math.}, 30(4):630--664, 1976.

\bibitem[KV95]{KnappVogan}
A.~W. Knapp and D.~A. Vogan, Jr.
\newblock {\em Cohomological induction and unitary representations}, volume~45
  of {\em Princeton Mathematical Series}.
\newblock Princeton University Press, Princeton, NJ, 1995.

\bibitem[LL58]{MR0093319}
L.~D. Landau and E.~M. Lifshitz.
\newblock {\em Quantum mechanics: non-relativistic theory. {C}ourse of
  {T}heoretical {P}hysics, {V}ol. 3}.
\newblock Addison-Wesley Series in Advanced Physics. Pergamon Press Ltd.,
  London-Paris; for U.S.A. and Canada: Addison-Wesley Publishing Co., Inc.,
  Reading, Mass;, 1958.
\newblock Translated from the Russian by J. B. Sykes and J. S. Bell.

\bibitem[Mau67]{maurin1967methods}
K.~Maurin.
\newblock {\em Methods of Hilbert spaces}.
\newblock Monografie matematyczne. Polish Scientific Publishers, 1967.

\bibitem[MPW13]{Miller13}
Willard Miller, Jr., Sarah Post, and Pavel Winternitz.
\newblock Classical and quantum superintegrability with applications.
\newblock {\em J. Phys. A}, 46(42):423001, 97, 2013.

\bibitem[Nai64]{Naimark}
M.~A. Naimark.
\newblock {\em Linear representations of the {L}orentz group}.
\newblock Translated by Ann Swinfen and O. J. Marstrand; translation edited by
  H. K. Farahat. A Pergamon Press Book. The Macmillan Co., New York, 1964.

\bibitem[Nel59]{Nelson59}
E.~Nelson.
\newblock Analytic vectors.
\newblock {\em Ann. of Math. (2)}, 70:572--615, 1959.

\bibitem[Pau26]{Pauli}
W.~Pauli.
\newblock {\"Uber das Wasserstoffspektrum vom Standpunkt der neuen
  Quantenmechanik}.
\newblock {\em Z. Phys.}, 36(5):336--363, 1926.

\bibitem[Per60]{MR0133586}
Arne Persson.
\newblock Bounds for the discrete part of the spectrum of a semi-bounded
  {S}chr\"odinger operator.
\newblock {\em Math. Scand.}, 8:143--153, 1960.

\bibitem[PP02]{Parfitt2002}
D.~G.~W. Parfitt and M.~E. Portnoi.
\newblock The two-dimensional hydrogen atom revisited.
\newblock {\em Journal of Mathematical Physics}, 43(10):4681--4691, 2002.

\bibitem[Sal61]{Saletan61}
E.~J. Saletan.
\newblock Contraction of {L}ie groups.
\newblock {\em J. Mathematical Phys.}, 2:1--21, 1961.

\bibitem[SBBM12]{Subag12}
E.~M. Subag, E.~M. Baruch, J.~L. Birman, and A.~Mann.
\newblock Strong contraction of the representations of the three-dimensional
  {L}ie algebras.
\newblock {\em J. Phys. A}, 45(26):265206, 25, 2012.

\bibitem[Sin05]{Singer}
Stephanie~Frank Singer.
\newblock {\em Linearity, symmetry, and prediction in the hydrogen atom}.
\newblock Undergraduate Texts in Mathematics. Springer, New York, 2005.

\bibitem[Tay86]{Taylor86}
Michael~E. Taylor.
\newblock {\em Noncommutative harmonic analysis}, volume~22 of {\em
  Mathematical Surveys and Monographs}.
\newblock American Mathematical Society, Providence, RI, 1986.

\bibitem[Tay96]{MR1395149}
Michael~E. Taylor.
\newblock {\em Partial differential equations. {II}: Qualitative studies of
  linear equations}, volume 116 of {\em Applied Mathematical Sciences}.
\newblock Springer-Verlag, New York, 1996.

\bibitem[TdCA98]{Torres1998}
G.~F. Torres~del Castillo and J.L.~Calvario Ac\'{o}cal.
\newblock {On the dynamical symmetry of the quantum Kepler problem}.
\newblock {\em Revista Mexicana de Fisica}, 44(4):344--352, 1998.

\bibitem[Vog81]{Vogan81}
David~A. Vogan, Jr.
\newblock {\em Representations of real reductive {L}ie groups}, volume~15 of
  {\em Progress in Mathematics}.
\newblock Birkh\"auser, Boston, Mass., 1981.

\bibitem[Wul11]{Wulfman}
Carl~E. Wulfman.
\newblock {\em Dynamical symmetry}.
\newblock World Scientific Publishing Co. Pte. Ltd., Hackensack, NJ, 2011.

\bibitem[YGC{\etalchar{+}}91]{PhysRevA.43.1186}
X.~L. Yang, S.~H. Guo, F.~T. Chan, K.~W. Wong, and W.~Y. Ching.
\newblock Analytic solution of a two-dimensional hydrogen atom. i.
  nonrelativistic theory.
\newblock {\em Phys. Rev. A}, 43:1186--1196, Feb 1991.

\end{thebibliography}
\bibliographystyle{alpha}

\end{document}